%% file: main.tex
\documentclass{article}
\usepackage{iclr2027_conference,times}
\iclrfinalcopy
\usepackage{etoolbox}
\makeatletter
\patchcmd{\@maketitle}{Published as a conference paper at ICLR 2027}{Preprint}{}{\PackageWarning{main}{Header patch failed}}
\makeatother
\usepackage[T1]{fontenc}
\usepackage[utf8]{inputenc}
\usepackage{amsmath,amssymb,booktabs,array,graphicx}
\usepackage{xcolor,arydshln}
\usepackage{subcaption}
\usepackage{placeins}
\usepackage{hyperref,url}
\usepackage{amsthm}
\usepackage{array}
\usepackage{xspace}
\newtheorem{proposition}{Proposition}[section]
\newtheorem{lemma}[proposition]{Lemma}
\newtheorem{corollary}[proposition]{Corollary}
\theoremstyle{definition}
\newtheorem{remark}[proposition]{Remark}
\hypersetup{colorlinks=true,citecolor=blue,linkcolor=blue,urlcolor=blue}
\newcommand{\E}{\mathbb{E}}
\newcommand{\Cov}{\operatorname{Cov}}
\newcommand{\Var}{\operatorname{Var}}
\newcommand{\KL}{D_{\mathrm{KL}}}
\newcommand{\ours}{Redwing\xspace}
\title{Tokens Change, Structure Endures:\\Spectral Watermarking for Generated Speech}
\author{Kanghwi Lee$^{1,2}$\thanks{Work done during an internship at NAVER Cloud.} \quad Kyeongseok Jeong$^{2}$ \quad Jeongmin Liu$^{2}$ \\
{\normalfont $^{1}$Institute of Neuroinformatics, University of Zurich and ETH Zurich, Switzerland} \\
{\normalfont $^{2}$NAVER Cloud} \\
\texttt{kanlee@ini.ethz.ch}, \texttt{\{ks.jeong, jeongmin.liu\}@navercorp.com}
}
\begin{document}
\raggedbottom
\maketitle
\begin{abstract}
Watermarking is a promising tool for establishing the provenance of AI-generated speech. While many neural audio watermarking methods rely on a separately trained watermark generator, token-level watermarking is a training-free alternative that operates directly during generation. Its main weakness is retokenization: decoding generated speech to a waveform and encoding it again can change token identities and erode the watermark.
To make the watermark robust to these changes, we propose \ours, REtokenization-Durable Watermarking IN Generation. It builds a graph from the token substitutions observed under retokenization, whose Laplacian yields a basis that assigns similar values to tokens likely to substitute for one another. Over this basis, embedding and detection functions are jointly optimized to preserve watermark signal through retokenization while limiting embedding distortion and detector variability on unwatermarked speech.
On the Moshi full-duplex system, after eight consecutive passes of Mimi resynthesis, \ours achieves 80.7\% TPR at a calibrated 1\% FPR, compared with 8.3\% for KGW and at most 7.3\% for WMAR. It also has the highest TPR after eight passes through three other neural codecs (77.5–93.0\%), and the gains generalize to TTS models at a speech-quality cost close to that of KGW.
These results show that retokenization is not merely a source of noise: its transition structure can be exploited as a design principle for robust token-level watermarking.
\end{abstract}
\input{sections/introduction}
\input{sections/related_work}
\input{sections/method}
\input{sections/experiments}
\input{sections/results}
\input{sections/discussion}

\input{sections/statements}
\bibliography{main}
\bibliographystyle{iclr2027_conference}
\clearpage
\appendix
\input{sections/appendix}
\end{document}

%% file: sections/introduction.tex
\section{Introduction}
Speech language models can now hold spoken conversations and read any text aloud in a chosen voice
\citep{borsos2023audiolm,wang2023valle,moshi,cosyvoice,moss}. The same ability makes it easy to impersonate a person or to fake a recording, so the
provider of such a model needs a way to recognize the speech that its model produced. Watermarking serves this
purpose by hiding a signal in the generated speech that only a detector holding a secret key can find. A watermark is
useful, however, only if it survives what happens to the audio after generation, whether platforms re-encode it
routinely or someone processes it on purpose to hide where it came from. We refer to any such processing as an
\emph{attack}. Resynthesis through a neural codec, which encodes the audio into discrete tokens and
decodes it again, is a particularly strong attack, because the decoder rebuilds the waveform from a compact code.

Most audio watermarks are embedded into the finished waveform by a trained network, and a second network detects
them \citep{wavmark,audioseal,timbre,singh2024silentcipher}. Most of them survive codecs that keep the waveform close to the
original, such as MP3 at 64\,kbps, but repeated passes through low-bitrate neural codecs remove them in most cases. Even CRAW \citep{craw}, which is trained to survive neural codecs, is removed by a
few passes through the Mimi codec in our experiments. Token-level watermarks take a different route. Introduced for text
with KGW \citep{kgw} and now deployed at scale \citep{dathathri2024synthid}, they bias the logits of the tokens that a language
model samples during generation with a secret key and detect this bias statistically, without training any network. Speech language models generate speech as the tokens of a neural codec
\citep{zeghidour2022soundstream,encodec}, so the same idea applies to them directly, and a watermark in the tokens is a
natural candidate to survive codec resynthesis, since the codec is built to preserve what its tokens encode.

For speech, however, token-level watermarks face a problem that text does not pose. A text detector reads the very tokens
that the model sampled, whereas a speech detector must encode the received waveform back into tokens. Even when nothing
happens to the audio, some of the tokens change in this round trip
(Figure~\ref{fig:overview_a}). Each further pass of resynthesis repeats the round trip, so a watermark that relies on the
exact identity of each token weakens with every pass.

Our watermark, \ours (REtokenization-Durable Watermarking IN Generation), is designed around these substitutions rather
than against them. The codec does not
replace tokens at random, and a token is mostly replaced by one of a few others that sound alike. We count these substitutions on a speech corpus and connect tokens that often replace each other in a graph
over the vocabulary. The eigenvectors of its graph Laplacian with the smallest eigenvalues change little between
such tokens, and we use them as a basis, so that a substitution changes the value of any combination of them little. In this basis, we solve for an embedding function,
which biases the model's choice of tokens during generation, and a separate detection function, which scores the tokens
recovered from the audio. We choose them so that as much of the watermark as possible survives the round trip, while it
changes the generated speech little and rarely fires on unwatermarked speech. The method needs only substitution counts and
statistics of unwatermarked generations, so it applies to a released model and codec unchanged.

We evaluate the watermark on the dialogue model Moshi and on two text-to-speech (TTS) models, CosyVoice3 and MOSS-TTS. So
that the comparison is fair, our watermark changes the model's sampling distribution by no more than the KGW baseline does. After eight
passes through Moshi's codec, Mimi, \ours is still detected in 80.7\,\% of Moshi's answers at a false-positive rate
of 1\,\%, whereas KGW falls to 8.3\,\% and WMAR \citep{wmar}, which fine-tunes the codec so that more tokens survive, to
7.3\,\% or less. \ours also has the highest detection rate after
most other low-bitrate neural codecs. To test whether the approach generalizes beyond dialogue models, we
rebuild it from the codec of each of two text-to-speech (TTS) models. After eight passes through that codec, it is detected in 84.7\,\% of
CosyVoice3's clips and 99.8\,\% of MOSS-TTS's, against 9.2\,\% and 6.2\,\% for KGW. Its cost in speech quality is about that of KGW.

Our contributions are thus a basis for token-level watermarks built from the substitutions of a codec round trip,
embedding and detection functions solved in this basis, and evidence that the approach generalizes across models and
codecs.

%% file: sections/related_work.tex
\section{Related work}
\paragraph{Post-hoc audio watermarks.}
WavMark \citep{wavmark}, AudioSeal \citep{audioseal}, Timbre \citep{timbre} and SilentCipher
\citep{singh2024silentcipher} train an embedding and a detection network jointly, with distortions applied between the
two, and CRAW \citep{craw} adds neural-codec resynthesis to these distortions. Latent-Mark \citep{latentmark} instead
shifts the codec latent of the waveform. Post-hoc watermarks apply to the output of any
generator, but their robustness varies widely across edits and compression \citep{liu2024audiomarkbench,wen2025sok}, and
repeated resynthesis through low-bitrate neural codecs removes most of it (Section~\ref{sec:results}).

\paragraph{Generation-time watermarks.}
These are embedded during generation. Latent watermarking
\citep{sanroman2024latent} trains the generator on watermarked audio. Token-level watermarks bias the logits of the sampled
tokens. KGW \citep{kgw} biases the logits of a keyed green list of tokens and detects an excess of green tokens. Later work makes
them distortion-free or undetectable \citep{kuditipudi2024robust,hu2024unbiased,christ2024undetectable}, studies their
robustness \citep{zhao2024provable,kirchenbauer2024reliability} and deploys them \citep{dathathri2024synthid}
\citep[see][]{liu2024survey}. For audio generation, Aligned-IS \citep{alignedis} is
distortion-free, and WMAR \citep{wmar} fine-tunes the codec so that more tokens survive re-encoding. HiPT \citep{hipt} is closest to our work. It applies
KGW to clusters of tokens that retokenization tends to confuse. Our basis also comes from these substitutions, but it
gives every token a real value that varies smoothly over the substitution graph, and our embedding and detection
functions differ from each other and are solved for, whereas KGW and HiPT use one random green list for both
(Appendix~\ref{app:related}).

\paragraph{Neural codecs and speech language models.}
Neural codecs such as SoundStream \citep{zeghidour2022soundstream}, EnCodec \citep{encodec}, DAC \citep{dac},
SpeechTokenizer \citep{speechtokenizer}, SNAC \citep{snac} and Mimi \citep{moshi} quantize speech into streams of tokens
with residual vector quantization or variants of it \citep{vandenoord2017vqvae}, and audio language models generate these tokens directly
\citep{borsos2023audiolm,wang2023valle,copet2023musicgen}.

\paragraph{Spectral bases on graphs.}
Our basis uses the eigenvectors of a graph Laplacian with the
smallest eigenvalues, which vary slowly across strongly connected nodes, as in spectral clustering, Laplacian eigenmaps
and graph signal processing
\citep{chung1997spectral,shi2000normalized,belkin2003laplacian,luxburg2007tutorial,shuman2013emerging}.

%% file: sections/method.tex
\providecommand{\syncW}{2}
\section{Methods}
\label{sec:method}
\begin{figure}[t]
\centering
\begin{subfigure}{\linewidth}
\centering
\includegraphics[width=\linewidth]{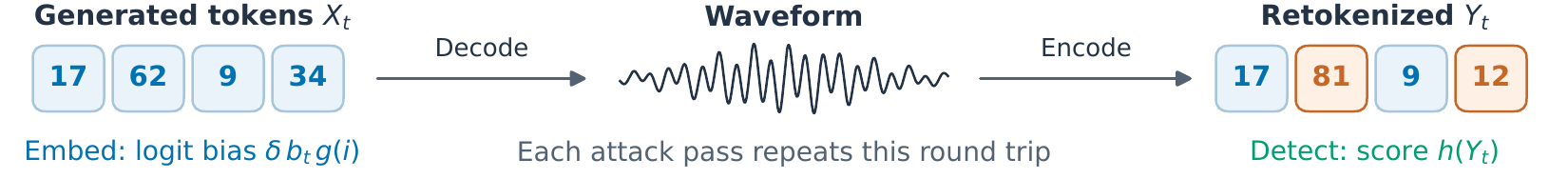}
\caption{Token identities change when the audio is decoded and encoded again.}
\label{fig:overview_a}
\end{subfigure}\\[4pt]
\begin{subfigure}{\linewidth}
\centering
\includegraphics[width=\linewidth]{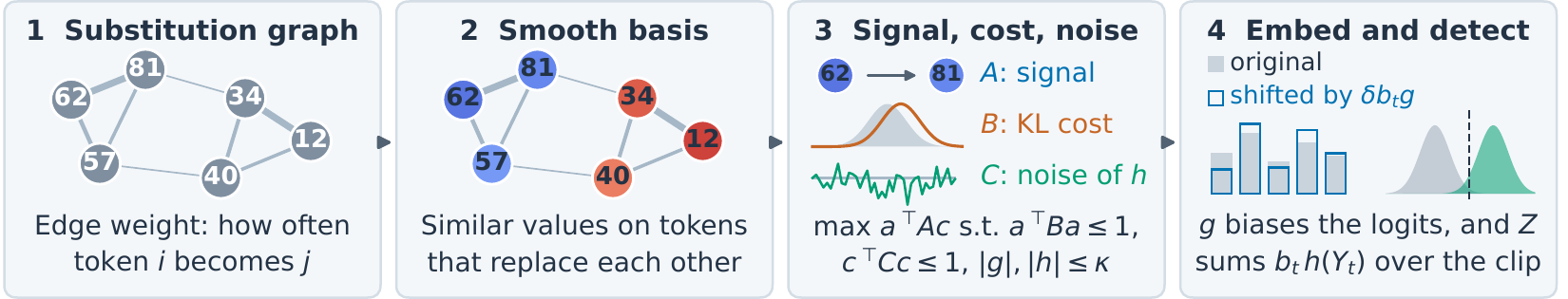}
\caption{Design the watermark around these substitutions.}
\label{fig:overview_b}
\end{subfigure}
\caption{\textbf{Overview.} (a) Some tokens change in every round trip, even without an attack. (b) We build a basis
from the observed substitutions and solve in it for the embedding and detection functions. Token IDs, graphs and
matrices are schematic.}
\label{fig:overview}
\end{figure}

The watermark is a keyed bias on the token logits of a speech language model, detected from the tokens recovered from the
audio. Appendix~\ref{app:method} gives derivations and implementation details.

\subsection{Setting}
\label{sec:setting}
\label{sec:detect}
A speech language model, whether it holds a dialogue or reads a text aloud (TTS), generates audio as the
tokens of a neural codec. Many speech codecs use residual vector quantization
(RVQ), in which several quantizers encode each frame in turn, and each quantizer level forms one \emph{stream} of tokens. At
frame $t$, the language model emits one token $X_{s,t}\in\{1,\ldots,V\}$ on each stream $s$, and the codec decoder turns these
tokens into a waveform. Encoding that waveform again gives tokens $Y_{s,t}$, which can differ from $X$ even when the audio is
unchanged. We call this re-encoding \emph{retokenization}, and the random map from $X$ to $Y$ the \emph{retokenization channel}.

The watermark is written on a subset $\mathcal S$ of the streams. A secret key defines a pseudorandom sign
$b_{s,t}\in\{-1,+1\}$ for every stream and frame. For each stream $s\in\mathcal S$, an \emph{embedding function}
$g_s\in\mathbb R^V$ assigns a real value to every token. During generation, the logits of stream $s$ at frame $t$ are shifted
before temperature scaling and sampling,
\begin{equation}
 \ell'_{s,t}(i)=\ell_{s,t}(i)+\delta\, b_{s,t}\, g_s(i),\qquad i\in\{1,\ldots,V\},
\label{eq:write}
\end{equation}
where $\delta\ge0$ is the watermark strength. Frames with $b_{s,t}=+1$ thus favor tokens with large $g_s$, and frames with
$b_{s,t}=-1$ those with small $g_s$.

The detector encodes any audio it receives with the codec encoder, whether or not it was generated, scores each resulting token $Y_{s,t}$ with a \emph{detection function} $h_s\in\mathbb R^V$, and checks whether the scores
rise and fall with the key signs,
\begin{equation}
 Z(\tau)=\frac{\sum_{s\in\mathcal S}\sum_t b_{s,t+\tau+d_s}\,h_s(Y_{s,t})}
 {\sqrt{\sum_{s\in\mathcal S}\sum_t h_s(Y_{s,t})^2}},\qquad
 Z^\star=\max_{|\tau|\le\syncW}Z(\tau),
\label{eq:detector}
\end{equation}
where $d_s$ is the generation delay of stream $s$, and the offset $\tau$ absorbs shifts of the
frame grid caused by retokenization. The numerator adds up how well the scores agree with the signs, and the denominator
scales this sum by the size of the scores. If the signs are modeled as independent fair signs, independent of the tokens, $Z(\tau)$
has mean zero and unit variance on any unwatermarked clip. On watermarked audio the agreement accumulates over frames, and $Z(\tau)$ is predicted
to grow with the square root of their number (Appendix~\ref{app:detect}). A useful watermark therefore needs an embedding function
that changes the sampling distribution little and a detection function that follows the signs after retokenization while
adding little noise on unwatermarked speech. Because $Z^\star$ is a maximum over offsets, it is not standard normal, so we
set its threshold empirically (Section~\ref{sec:experiments}).

\subsection{A basis from token substitutions}
\label{sec:basis}
We construct $g_s$ and $h_s$ in the same way for every stream and omit the index $s$ from here on. We start from the substitutions
that retokenization makes on speech recordings.
Each recording is encoded with the codec of the language model (Mimi for Moshi),
decoded, and encoded again. After aligning the two token
sequences in time (Appendix~\ref{app:counts}), we count $N_{ij}$, the number of frames in which token $i$ was recovered as
token $j$. We keep the tokens observed at least 50 times, which we call the \emph{support}, and connect them in the
\emph{substitution graph} by how often either replaced the other,
\begin{equation}
 W_{ij}=\tfrac12\left(N_{ij}+N_{ji}\right)\mathbf 1\{i\ne j\},\qquad
 D_{ii}=\textstyle\sum_j W_{ij},\qquad
 L=I-D^{-1/2}WD^{-1/2}.
\label{eq:graph}
\end{equation}
Here $L$ is the normalized Laplacian of the substitution graph \citep{chung1997spectral,luxburg2007tutorial}. We exclude self-transitions, because with them
the leading modes become indicators of single tokens that tend to survive, which carry little of the watermark
(Appendix~\ref{app:basis}).

Let $u_k$ be a unit eigenvector of $L$ with eigenvalue $\lambda_k$. The rescaled vector $\phi_k=D^{-1/2}u_k$, set to zero outside
the support, is a function on the vocabulary that satisfies
\begin{equation}
 \lambda_k=\tfrac12\textstyle\sum_{i,j}W_{ij}\big(\phi_k(i)-\phi_k(j)\big)^2
\label{eq:rayleigh}
\end{equation}
(Appendix~\ref{app:basis}). The eigenvalue is thus a weighted sum of squared differences between the values of $\phi_k$ on
tokens that replace each other. A function with small $\lambda_k$ takes similar values on such tokens, so retokenization
changes it little. By design, we discard the constant functions, which have $\lambda=0$ and
do not vary within connected tokens, and keep the functions $\phi_1,\ldots,\phi_K$ of the next $K=16$ eigenvalues as the
\emph{basis} (Appendix~\ref{app:hyper}). Every linear combination of these functions is smooth in the same sense, because its
weighted sum of squared differences is at most $\lambda_K$ times its squared size (Appendix~\ref{app:basis}). We write
$\phi(i)=(\phi_1(i),\ldots,\phi_K(i))\in\mathbb R^K$ for the basis values of token $i$.

\subsection{Signal, cost, and noise}
\label{sec:moments}
Let $p_t$ be the sampling distribution of the language model at frame $t$. How well a pair of functions works depends on
three properties of the basis, which three $K\times K$ matrices measure: how its values carry over from a generated token
to the token that the detector recovers ($A$), how much they vary among the tokens that the language model is likely to
sample ($B$), and how they are distributed on unwatermarked speech ($C$).

\paragraph{Surviving signal $A$.} A shift in the generated tokens helps the detector only if it is still present in the
recovered tokens. Let the \emph{transition matrix} $P_{ij}=N_{ij}/\sum_kN_{ik}$ be the probability that token $i$ is recovered as token $j$.
Unlike the substitution graph, $P$ includes $j=i$, since a token that survives is recovered as itself. The expected
basis values of the token recovered from token $i$ are $(P\phi)(i)=\sum_jP_{ij}\phi(j)$, and
\begin{equation}
 A=\E_t\Big[\Cov_{X\sim p_t}\big(\phi(X),(P\phi)(X)\big)\Big]
\label{eq:A}
\end{equation}
is the covariance between the basis values of the generated token and the expected basis values of the recovered token,
averaged over frames.

\paragraph{Write cost $B$.} Every change to the sampling distribution moves the generated speech away from what the model
would have produced, so we measure the price of a bias by the KL divergence it causes, the same quantity that we hold
equal across methods (Section~\ref{sec:experiments}). A bias can change the distribution only through the tokens that
the language model is likely to sample. A function that takes the same value on all of them leaves the distribution
unchanged and costs nothing, and the more it varies among them, the higher its cost. The matrix
\begin{equation}
 B=\E_t\Big[\Cov_{X\sim p_t}\big(\phi(X)\big)\Big]
\label{eq:B}
\end{equation}
measures this variation for the basis, averaged over frames. Section~\ref{sec:solve} shows that, to first order, it
determines the KL divergence of the bias.

\paragraph{Null variance $C$.} Let $H_0$ denote unwatermarked speech, generated by the same language model without the bias,
and $Y$ the tokens that the detector obtains from it. The mean and covariance of their basis values are
\begin{equation}
 \mu_0=\E_{H_0}\big[\phi(Y)\big],\qquad C=\Cov_{H_0}\big(\phi(Y)\big).
\label{eq:abc}
\end{equation}
On unwatermarked speech every term in the numerator of Eq.~\eqref{eq:detector} is noise, and $C$ determines how large
it is (Section~\ref{sec:solve}). We estimate all four quantities from unwatermarked generations for a separate set of fitting prompts
(Section~\ref{sec:experiments}).

\subsection{Embedding and detection functions}
\label{sec:solve}
Both functions are linear combinations of the basis,
\begin{equation}
 g(i)=a^\top\phi(i),\qquad h(i)=c^\top\big(\phi(i)-\mu_0\big),\qquad a,c\in\mathbb R^K.
\label{eq:scores}
\end{equation}
Every such function inherits the smoothness of the basis (Section~\ref{sec:basis}), so the design reduces to choosing the
coefficient vectors $a$ and $c$. The detection function subtracts $\mu_0$ so that it has zero mean on unwatermarked speech
(Appendix~\ref{app:firstorder}).

\paragraph{Objective.} The FPR threshold is set on unwatermarked audio, so a pair of embedding and detection functions is better the more it raises
the mean of $Z$ on watermarked audio at a given cost. Write $\beta=\delta/T$ for sampling temperature $T$. To leading order in
$\beta$, the signal, cost and noise of one frame are quadratic forms in the matrices of Section~\ref{sec:moments}
(Appendix~\ref{app:firstorder}),
\begin{equation}
 \E\big[b_th(Y_t)\big]\approx\beta\,a^\top Ac,\qquad
 \E_t\big[\KL(q_t\Vert p_t)\big]\approx\tfrac12\beta^2\,a^\top Ba,\qquad
 \E_{H_0}\big[h(Y)^2\big]=c^\top Cc,
\label{eq:signal}
\end{equation}
In order, these are the mean contribution of one frame to the numerator of Eq.~\eqref{eq:detector}, the KL divergence
of the biased sampling distribution $q_t$ from $p_t$, and the mean squared detector score on unwatermarked speech. 
A clip of $n$ frames and a fixed cost at
$\varepsilon$ nats per frame gives (Appendix~\ref{app:objective})
\begin{equation}
 \E[Z]\approx\sqrt{2\varepsilon n}\;J(a,c),\qquad J(a,c)=\frac{a^\top Ac}{\sqrt{a^\top Ba}\,\sqrt{c^\top Cc}}.
\label{eq:ratio}
\end{equation}
With $\varepsilon$ and $n$ fixed, this approximate $\E[Z]$ depends only on $J$, which depends only on $a$ and $c$. We
therefore choose the coefficients that maximize $J$, the signal per unit
of cost and of noise. They are $a=B^{-1/2}u$ and $c=C^{-1/2}v$, where $u$ and $v$ are the leading singular vectors of
$B^{-1/2}AC^{-1/2}$, as in canonical correlation analysis \citep{hotelling1936relations} (Appendix~\ref{app:objective}).

\paragraph{Bounded problem.} The maximizer of $J$ gives extreme values to a few rarely sampled tokens, which $B$ and $C$
barely penalize. When such a token becomes a candidate, the bias forces or excludes it, and speech quality collapses
before the watermark becomes detectable (Appendix~\ref{app:bounded}). We therefore solve a surrogate that caps the cost and the
noise and bounds the value of every token,
\begin{equation}
 \max_{a,c}\ a^\top Ac\quad\text{subject to}\quad a^\top Ba\le1,\quad c^\top Cc\le1,\quad
 \big|g(i)\big|,\big|h(i)\big|\le\kappa\ \text{for all tokens } i.
\label{eq:solve}
\end{equation}
We use $\kappa=5$ (Appendix~\ref{app:hyper}). For fixed $c$ the problem is a second-order cone program \citep{boyd2004convex} in $a$, and vice versa, so we
alternate between the two (Appendix~\ref{app:bounded}).

%% file: sections/experiments.tex
\section{Experimental setup}
\label{sec:experiments}

\paragraph{Language models and data.}
Our primary language model is Moshi \citep{moshi}, a full-duplex dialogue model that generates eight streams of Mimi
tokens, each with a vocabulary of 2,048, at 12.5 frames per second. We prompt it with 1,000 spoken questions from WildVoice
\citep{voicebench}, split into a fitting set of 150 prompts for the statistics of Section~\ref{sec:moments}, a validation
set of 250 for choosing the strength $\delta$ and the watermarked streams, and a test set of 600 for all reported results.
To test whether the construction transfers, we also watermark two TTS language models, CosyVoice3
\citep{cosyvoice} and MOSS-TTS \citep{moss}, which read 250 validation and 600 test texts in the voices of WildVoice
speakers. Each has its own basis and functions, built from its own codec. For all three, the substitutions are counted
on 20 hours of LibriSpeech \citep{librispeech} (Appendix~\ref{app:protocol}). We watermark streams 1--4 of Moshi and
1--16 of the 32 of MOSS-TTS.

\paragraph{Baselines.}
KGW \citep{kgw} runs on the same streams as our method, with $\delta=2$ and one green list of a quarter of the
vocabulary for all frames, as in WMAR \citep{wmar}. WMAR keeps this watermark and fine-tunes Mimi so that more tokens
survive retokenization. We use its released codecs, fine-tuned with (FT+Augs) or without (FT) augmentations, so it
appears only on Moshi. The four post-hoc methods, AudioSeal \citep{audioseal}, WavMark \citep{wavmark}, Timbre
\citep{timbre} and CRAW \citep{craw}, watermark the unwatermarked outputs of each language model with their released weights. We
also tested HiPT \citep{hipt}, but at its nominal 1\,\% level it flags almost all unwatermarked Moshi answers, which its
authors confirmed, so we omit it. Aligned-IS \citep{alignedis}, which we ported to Moshi, is discussed in
Appendix~\ref{app:alignedis}.

\paragraph{Matching the cost.}
A stronger watermark is easier to detect but changes the speech more, so we compare methods at the same cost, the KL
divergence of the biased distribution $q_t$ from $p_t$, summed over the watermarked streams and averaged over the
validation frames. For a fair comparison, we set the budget, how far our $q_t$ may
diverge from $p_t$, to the cost of KGW with $\delta=2$ on the same streams. We use the largest $\delta$ on a validation grid whose cost stays within it, which gives $\delta=0.7$ on Moshi and
$\delta=0.9$ on both TTS models (Figures~\ref{fig:moshi_b} and~\ref{fig:moshi_c}, Table~\ref{tab:quality}). We select
$\delta$ from this measured cost, since its first-order prediction underestimates it on the TTS models
(Appendix~\ref{app:protocol}).

\paragraph{Attacks.}
The detector always encodes the received audio with the codec of the language model, so even audio that
arrives as generated, which we call the identity condition, is read as the retokenized sequence $Y$ rather than the
generated tokens $X$. An attack applies up to eight passes, each of which decodes and re-encodes the waveform, either with
the language model's own codec (native passes) or with EnCodec at 6\,kbps \citep{encodec}, SpeechTokenizer
\citep{speechtokenizer}, SNAC \citep{snac}, the 16\,kHz DAC model \citep{dac} and, on the TTS models, Mimi (foreign
passes). High-fidelity codecs and signal-processing attacks are in Appendix~\ref{app:codecs} and~\ref{app:dsp}.

\paragraph{Thresholds and error rates.}
Each detector's threshold is set so that at most 1\,\% of a calibration set of unwatermarked answers, used for
nothing else, exceed it (Appendix~\ref{app:data}). It is not recalibrated per attack, so every true-positive rate (TPR)
is at a false-positive rate (FPR) of at most 1\,\% on the calibration answers, and the held-out FPR is in
Appendix~\ref{app:fpr}. Intervals for TPR and FPR are 95\,\% Wilson score intervals
\citep{wilson1927probable}. In the tables, bold marks the best value in each column of a
model and every value whose interval overlaps it, and nothing when all overlap.

\paragraph{Speech quality.}
We score quality with UTMOSv2 \citep{utmosv2}, DNSMOS Pro \citep{dnsmospro} and the general and TTS versions of NISQA
\citep{nisqa}. Moshi sometimes gives an empty answer, with no word in its text stream, which the predictors rate as
very poor speech, so we score Moshi's quality on non-empty answers only, while detection uses all answers (Appendix~\ref{app:nonempty}). The TTS models, which read a given text, give no
empty answers.

%% file: sections/results.tex
\section{Results}
\label{sec:results}

\begin{figure}[t]
\centering
\includegraphics[width=\linewidth]{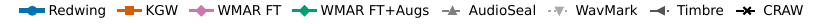}\\[2pt]
\begin{subfigure}[t]{0.41\linewidth}
\centering
\includegraphics[width=\linewidth]{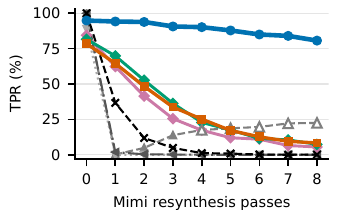}
\caption{Repeated Mimi resynthesis}
\label{fig:moshi_a}
\end{subfigure}\hfill
\begin{subfigure}[t]{0.28\linewidth}
\centering
\includegraphics[width=\linewidth]{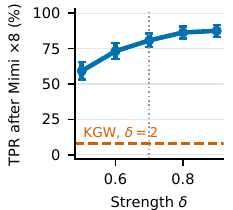}
\caption{Detection vs.\ $\delta$}
\label{fig:moshi_b}
\end{subfigure}\hfill
\begin{subfigure}[t]{0.28\linewidth}
\centering
\includegraphics[width=\linewidth]{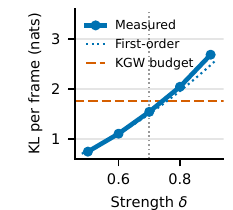}
\caption{Cost vs.\ $\delta$}
\label{fig:moshi_c}
\end{subfigure}
\caption{\textbf{Detection and cost on Moshi.} TPR at the fixed threshold (calibrated FPR $\le 1\,\%$). (a) On the test set,
where pass 0 is the identity condition. Open markers show attacks that flag more than 5\,\% of unwatermarked answers. (b, c) On
the validation set. Dashed lines show KGW, whose cost is the budget. Dotted lines mark the selected $\delta=0.7$ and,
in (c), the first-order prediction.}
\label{fig:moshi}
\end{figure}

\begin{table}[t]
\caption{\textbf{Detection on Moshi after eight codec passes.} Values are TPR (\%) at the fixed threshold $\pm$ the
half-width of the 95\,\% interval. Post-hoc methods are below the line, and bold is as defined in
Section~\ref{sec:experiments}. $^\ast$The same attack also flags more than 5\,\% of unwatermarked answers
(Table~\ref{tab:fpr}), so the cell is not bold. WavMark, which detects 0\,\% after every codec here, is omitted
(Appendix~\ref{app:codecs}).}
\label{tab:moshi_codecs}
\centering\footnotesize\setlength{\tabcolsep}{3pt}
\input{tables/r_t1_moshi_codecs}
\end{table}

\paragraph{Detection on Moshi.}
Figure~\ref{fig:moshi_a} follows detection through repeated Mimi resynthesis. Our watermark keeps 80.7\,\% after eight
passes, while KGW falls to 8.3\,\%. WMAR does no better and ends at 7.3\,\% or below. The post-hoc methods start near 100\,\%, but one Mimi pass removes almost all
of their watermark, and CRAW, which lasts longest, is gone after four passes. AudioSeal drops to 0.2\,\% after one pass
and climbs back to 22.5\,\% after eight, as each pass raises its score on all audio toward its fixed threshold, which then also
flags 10.5\,\% of unwatermarked answers, against 0.5\,\% for ours (Appendix~\ref{app:fpr}). At this
false-positive rate, AudioSeal is far from usable.

Our watermark also has the highest TPR after eight passes through EnCodec, SpeechTokenizer and SNAC
(Table~\ref{tab:moshi_codecs}). AudioSeal's 82.3\,\% after EnCodec says little about its watermark, since its threshold
then also flags 75.3\,\% of unwatermarked answers. Only on DAC16 does a post-hoc method, CRAW, do better
than ours, which still beats every token-domain baseline there.

\begin{figure}[t]
\centering
\begin{subfigure}[t]{0.38\linewidth}
\centering
\includegraphics[width=\linewidth]{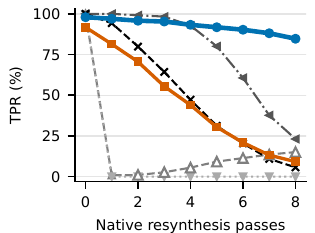}
\caption{CosyVoice3}
\label{fig:tts_a}
\end{subfigure}\hspace{0.02\linewidth}
\begin{subfigure}[t]{0.38\linewidth}
\centering
\includegraphics[width=\linewidth]{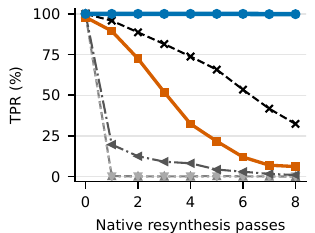}
\caption{MOSS-TTS}
\label{fig:tts_b}
\end{subfigure}\hspace{0.02\linewidth}
\begin{minipage}[t]{0.17\linewidth}
\vspace{-1.35in}
\includegraphics[width=\linewidth]{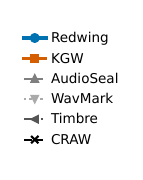}
\end{minipage}
\caption{\textbf{Repeated native resynthesis on the TTS models.} TPR on the test sets, as in Figure~\ref{fig:moshi_a}.}
\label{fig:tts}
\end{figure}

\begin{table}[t]
\caption{\textbf{Detection on the TTS models after eight codec passes.} The layout is as in Table~\ref{tab:moshi_codecs}, with each
model's own resynthesis (Native) as the first codec.}
\label{tab:tts_codecs}
\centering\footnotesize\setlength{\tabcolsep}{3pt}
\input{tables/r_t4_tts_codecs}
\end{table}

\paragraph{Transfer to text-to-speech models.}
Rebuilt from each model's own codec, the construction carries over to both TTS models (Figure~\ref{fig:tts}). After
eight native passes, our watermark keeps 84.7\,\% on CosyVoice3 and 99.8\,\% on MOSS-TTS, against 9.2\,\% and 6.2\,\%
for KGW and at most 32.4\,\% for the post-hoc methods. It also has the highest TPR after every foreign codec
(Table~\ref{tab:tts_codecs}). Its weakest case is Mimi on MOSS-TTS, but there no other method exceeds
10\,\%.

\begin{table}[t]
\caption{\textbf{Speech quality at the selected strengths.} Mean $\pm$ 95\,\% CI (higher is better), for Moshi on
non-empty answers. KL/frame is the cost. $^\dagger$Decoded with WMAR's fine-tuned decoder, which changes quality even
without a watermark. Bold (Section~\ref{sec:experiments}) covers only the watermarked rows.}
\label{tab:quality}
\centering\footnotesize\setlength{\tabcolsep}{3pt}
\input{tables/r_t2_quality}
\end{table}

\paragraph{Speech quality.}
Our watermark costs about as much quality as KGW (Table~\ref{tab:quality}). On Moshi, it lowers
UTMOSv2 by 0.10, as much as KGW does, and keeps the highest DNSMOS Pro and NISQAv2 among the watermarked rows. It
changes every predictor by at most 0.01 on CosyVoice3 and lowers UTMOSv2 by 0.12 on MOSS-TTS, against 0.09 for KGW. CRAW, the only post-hoc
method that survives some codecs, lowers UTMOSv2 by more than 0.2 on every model
(Appendix~\ref{app:quality}).

\paragraph{Ablations.}
Removing one component at a time shows where the robustness comes from (Table~\ref{tab:ablation} in
Appendix~\ref{app:ablation}). After eight Mimi passes, a random basis under the same bounded solve keeps 52.3\,\%, and a
basis from the transition matrix $P$ keeps 68.5\,\%, against 80.7\,\% for ours. The solved functions matter as well. A
fixed random function keeps only 16.8\,\%, and detecting with the embedding function in place of the detection function
keeps 76.5\,\%. Without the amplitude bounds, the cost grows so steeply with $\delta$ that the budget allows only
$\delta=0.07$, where the watermark keeps 18.7\,\% and UTMOSv2 falls well below that of every other variant. DAC16 is the
exception, where both alternative bases and detection with $g$ do better than ours.

%% file: tables/r_t1_moshi_codecs.tex
\begin{tabular}{lrrrrr}
\toprule
Method & Mimi (native) & EnCodec & SpeechTok. & SNAC & DAC16 \\
\midrule
Redwing & \textbf{80.7}\,{\scriptsize$\pm$3.2} & \textbf{85.0}\,{\scriptsize$\pm$2.9} & \textbf{77.5}\,{\scriptsize$\pm$3.3} & \textbf{93.0}\,{\scriptsize$\pm$2.1} & 34.2\,{\scriptsize$\pm$3.8} \\
KGW & 8.3\,{\scriptsize$\pm$2.2} & 16.5\,{\scriptsize$\pm$3.0} & 5.3\,{\scriptsize$\pm$1.8} & 63.2\,{\scriptsize$\pm$3.8} & 12.5\,{\scriptsize$\pm$2.6} \\
WMAR FT & 5.5\,{\scriptsize$\pm$1.8} & 24.7\,{\scriptsize$\pm$3.4} & 2.2\,{\scriptsize$\pm$1.2} & 58.5\,{\scriptsize$\pm$3.9} & 10.3\,{\scriptsize$\pm$2.4} \\
WMAR FT+Augs & 7.3\,{\scriptsize$\pm$2.1} & 15.7\,{\scriptsize$\pm$2.9} & 2.0\,{\scriptsize$\pm$1.2} & 57.2\,{\scriptsize$\pm$3.9} & 20.7\,{\scriptsize$\pm$3.2} \\
\midrule
AudioSeal & 22.5$^\ast$\,{\scriptsize$\pm$3.3} & 82.3$^\ast$\,{\scriptsize$\pm$3.0} & 1.5\,{\scriptsize$\pm$1.0} & 0.0\,{\scriptsize$\pm$0.3} & 0.2\,{\scriptsize$\pm$0.5} \\
Timbre & 0.0\,{\scriptsize$\pm$0.3} & 0.3\,{\scriptsize$\pm$0.6} & 0.0\,{\scriptsize$\pm$0.3} & 0.0\,{\scriptsize$\pm$0.3} & 0.0\,{\scriptsize$\pm$0.3} \\
CRAW & 0.2\,{\scriptsize$\pm$0.5} & 38.8\,{\scriptsize$\pm$3.9} & 1.2\,{\scriptsize$\pm$0.9} & 4.0\,{\scriptsize$\pm$1.6} & \textbf{47.7}\,{\scriptsize$\pm$4.0} \\
\bottomrule
\end{tabular}

%% file: tables/r_t4_tts_codecs.tex
\begin{tabular}{lrrrrrr}
\toprule
Method & Native & Mimi & EnCodec & SpeechTok. & SNAC & DAC16 \\
\midrule
\multicolumn{7}{l}{\textbf{CosyVoice3}} \\
Redwing & \textbf{84.7}\,{\scriptsize$\pm$2.9} & \textbf{68.8}\,{\scriptsize$\pm$3.7} & \textbf{89.6}\,{\scriptsize$\pm$2.5} & \textbf{61.2}\,{\scriptsize$\pm$3.9} & \textbf{91.9}\,{\scriptsize$\pm$2.2} & \textbf{94.0}\,{\scriptsize$\pm$1.9} \\
KGW & 9.2\,{\scriptsize$\pm$2.3} & 1.2\,{\scriptsize$\pm$0.9} & 9.9\,{\scriptsize$\pm$2.4} & 2.8\,{\scriptsize$\pm$1.4} & 29.3\,{\scriptsize$\pm$3.6} & 21.9\,{\scriptsize$\pm$3.3} \\
\midrule
AudioSeal & 15.1$^\ast$\,{\scriptsize$\pm$2.9} & 30.0$^\ast$\,{\scriptsize$\pm$3.7} & 56.4$^\ast$\,{\scriptsize$\pm$4.0} & 0.3\,{\scriptsize$\pm$0.6} & 0.0\,{\scriptsize$\pm$0.3} & 0.2\,{\scriptsize$\pm$0.5} \\
Timbre & 23.1\,{\scriptsize$\pm$3.4} & 0.3\,{\scriptsize$\pm$0.6} & 3.0\,{\scriptsize$\pm$1.4} & 0.0\,{\scriptsize$\pm$0.3} & 0.0\,{\scriptsize$\pm$0.3} & 0.2\,{\scriptsize$\pm$0.5} \\
CRAW & 5.7\,{\scriptsize$\pm$1.9} & 1.2\,{\scriptsize$\pm$0.9} & 61.5\,{\scriptsize$\pm$3.9} & 10.7\,{\scriptsize$\pm$2.5} & 11.7\,{\scriptsize$\pm$2.6} & 69.7\,{\scriptsize$\pm$3.7} \\
\midrule[\heavyrulewidth]\addlinespace[2pt]
\multicolumn{7}{l}{\textbf{MOSS-TTS}} \\
Redwing & \textbf{99.8}\,{\scriptsize$\pm$0.5} & \textbf{35.9}\,{\scriptsize$\pm$3.9} & \textbf{95.1}\,{\scriptsize$\pm$1.8} & \textbf{65.1}\,{\scriptsize$\pm$3.8} & \textbf{93.1}\,{\scriptsize$\pm$2.1} & \textbf{83.5}\,{\scriptsize$\pm$3.0} \\
KGW & 6.2\,{\scriptsize$\pm$1.9} & 0.3\,{\scriptsize$\pm$0.6} & 0.7\,{\scriptsize$\pm$0.7} & 0.7\,{\scriptsize$\pm$0.7} & 3.8\,{\scriptsize$\pm$1.6} & 2.2\,{\scriptsize$\pm$1.2} \\
\midrule
AudioSeal & 0.0\,{\scriptsize$\pm$0.3} & 9.3\,{\scriptsize$\pm$2.3} & 44.1$^\ast$\,{\scriptsize$\pm$4.0} & 0.2\,{\scriptsize$\pm$0.5} & 0.0\,{\scriptsize$\pm$0.3} & 0.5\,{\scriptsize$\pm$0.6} \\
Timbre & 1.0\,{\scriptsize$\pm$0.9} & 0.2\,{\scriptsize$\pm$0.5} & 1.8\,{\scriptsize$\pm$1.1} & 0.0\,{\scriptsize$\pm$0.3} & 0.0\,{\scriptsize$\pm$0.3} & 0.0\,{\scriptsize$\pm$0.3} \\
CRAW & 32.4\,{\scriptsize$\pm$3.7} & 0.2\,{\scriptsize$\pm$0.5} & 23.0\,{\scriptsize$\pm$3.4} & 0.7\,{\scriptsize$\pm$0.7} & 1.5\,{\scriptsize$\pm$1.0} & 33.9\,{\scriptsize$\pm$3.8} \\
\bottomrule
\end{tabular}

%% file: tables/r_t2_quality.tex
\begin{tabular}{lrrrrr}
\toprule
Method & UTMOSv2 & DNSMOS Pro & NISQAv2 & NISQA-TTS & KL/frame \\
\midrule
\multicolumn{6}{l}{\textbf{Moshi (non-empty answers)}} \\
Unwatermarked & 3.15\,{\scriptsize$\pm$0.03} & 4.38\,{\scriptsize$\pm$0.04} & 4.64\,{\scriptsize$\pm$0.04} & 3.56\,{\scriptsize$\pm$0.05} & 0 \\
Unwatermarked, FT decoder$^\dagger$ & 3.04\,{\scriptsize$\pm$0.03} & 4.33\,{\scriptsize$\pm$0.04} & 4.55\,{\scriptsize$\pm$0.04} & 3.40\,{\scriptsize$\pm$0.04} & 0 \\
Unwatermarked, FT+Augs decoder$^\dagger$ & 3.16\,{\scriptsize$\pm$0.03} & 4.36\,{\scriptsize$\pm$0.04} & 4.61\,{\scriptsize$\pm$0.04} & 3.54\,{\scriptsize$\pm$0.05} & 0 \\
Redwing & \textbf{3.05}\,{\scriptsize$\pm$0.03} & \textbf{4.47}\,{\scriptsize$\pm$0.02} & \textbf{4.62}\,{\scriptsize$\pm$0.03} & \textbf{3.52}\,{\scriptsize$\pm$0.04} & 1.52 \\
KGW & \textbf{3.06}\,{\scriptsize$\pm$0.04} & 4.35\,{\scriptsize$\pm$0.05} & 4.51\,{\scriptsize$\pm$0.06} & \textbf{3.44}\,{\scriptsize$\pm$0.05} & 1.73 \\
WMAR FT$^\dagger$ & 2.97\,{\scriptsize$\pm$0.04} & 4.30\,{\scriptsize$\pm$0.05} & 4.43\,{\scriptsize$\pm$0.05} & 3.27\,{\scriptsize$\pm$0.05} & 1.73 \\
WMAR FT+Augs$^\dagger$ & \textbf{3.04}\,{\scriptsize$\pm$0.04} & 4.32\,{\scriptsize$\pm$0.05} & 4.46\,{\scriptsize$\pm$0.06} & 3.41\,{\scriptsize$\pm$0.05} & 1.73 \\
\midrule[\heavyrulewidth]\addlinespace[2pt]
\multicolumn{6}{l}{\textbf{CosyVoice3}} \\
Unwatermarked & 2.91\,{\scriptsize$\pm$0.03} & 4.30\,{\scriptsize$\pm$0.03} & 3.78\,{\scriptsize$\pm$0.07} & 3.52\,{\scriptsize$\pm$0.05} & 0 \\
Redwing & 2.91\,{\scriptsize$\pm$0.03} & 4.29\,{\scriptsize$\pm$0.03} & 3.77\,{\scriptsize$\pm$0.07} & 3.52\,{\scriptsize$\pm$0.05} & 0.60 \\
KGW & 2.91\,{\scriptsize$\pm$0.03} & 4.29\,{\scriptsize$\pm$0.03} & 3.73\,{\scriptsize$\pm$0.07} & 3.50\,{\scriptsize$\pm$0.05} & 0.61 \\
\midrule[\heavyrulewidth]\addlinespace[2pt]
\multicolumn{6}{l}{\textbf{MOSS-TTS}} \\
Unwatermarked & 2.85\,{\scriptsize$\pm$0.04} & 4.20\,{\scriptsize$\pm$0.03} & 3.61\,{\scriptsize$\pm$0.07} & 3.28\,{\scriptsize$\pm$0.05} & 0 \\
Redwing & 2.73\,{\scriptsize$\pm$0.04} & 4.16\,{\scriptsize$\pm$0.03} & 3.60\,{\scriptsize$\pm$0.07} & 3.24\,{\scriptsize$\pm$0.05} & 2.44 \\
KGW & 2.77\,{\scriptsize$\pm$0.03} & 4.19\,{\scriptsize$\pm$0.03} & 3.65\,{\scriptsize$\pm$0.06} & 3.27\,{\scriptsize$\pm$0.05} & 2.46 \\
\bottomrule
\end{tabular}

%% file: sections/discussion.tex
\section{Discussion and limitations}

\paragraph{Why the watermark survives.}
Three properties of the design explain the robustness. First, the basis gives similar values to tokens that replace each
other, so a substituted token still carries most of its value, even after several passes. With a green list, a
substituted token is green only by chance, and every pass loses the tokens it changes. Second, the embedding and detection
functions are solved separately, since one function must shift the sampling distribution cheaply and the other must pick up what
survives with little noise on unwatermarked speech (Corollary~\ref{cor:matched}). Third, the amplitude bounds
spread the bias over many tokens, which keeps the cost moderate. Removing any of the three lowers the TPR after eight Mimi passes
(Table~\ref{tab:ablation}).

\paragraph{Transfer to other codecs.}
The functions are fitted to the round trip of the language model's own codec, yet after eight passes through a foreign
low-bitrate codec our TPR stays above 50\,\% in 12 of the 14 pairs of model and codec in Tables~\ref{tab:moshi_codecs}
and~\ref{tab:tts_codecs}. This requires that a foreign pass, followed by re-encoding with the model's codec, mostly
replaces each token by one of its neighbors in the native substitution graph. Counting the substitutions of a foreign pass
would test this and explain the two exceptions.

\paragraph{Token-domain and post-hoc watermarks are complementary.}
Post-hoc methods, which read the waveform directly, are detected almost perfectly without attack, survive codecs that
preserve the waveform closely, and are unaffected by time shifts and cropping (Appendix~\ref{app:codecs}
and~\ref{app:dsp}). Repeated passes through low-bitrate neural codecs, however, remove most of their watermark, and pitch
shifts remove all of it. Ours shows the opposite pattern. It lives in the tokens that such codecs reproduce and
survives neural resynthesis and pitch shifts, but not sub-frame shifts or cropping. A system could embed both
watermarks, and how they interact is left for future work. Comparisons of the two should also report the FPR on attacked unwatermarked
audio, a large share of which AudioSeal's threshold flags after some codecs.

\paragraph{Leaving the model unchanged.}
WMAR adapts the codec to the watermark, but under repeated Mimi resynthesis it does no better than KGW, and its
fine-tuned decoder changes quality even without a watermark (Table~\ref{tab:quality}). \ours instead adapts the
watermark to the released codec, and the recipe carries over to two TTS models at a quality cost close to KGW's.

\paragraph{Limitations.}
Our watermark is strongest under resynthesis through low-bitrate neural codecs, and post-hoc methods are the better
choice for codecs that keep the waveform nearly intact. Short answers carry too few frames. Almost
all of our misses without attack, and 72 of 116 after eight Mimi passes, come from answers of at most 4\,s
(Appendix~\ref{app:duration}). The offset search recovers shifts of whole frames, but sub-frame shifts, changes of speed
and cropping misalign the frame grid and remove most of the watermark, as for every token-level method
(Appendix~\ref{app:dsp}). An untested detector that encodes the
audio at several sub-frame offsets could recover fractional shifts, though not speed changes, at the price of more
computation and a higher threshold. Finally, the functions are fitted to one codec, and counting
substitutions over several codecs may make the basis more robust to foreign ones.

\section{Conclusion}
Token-level watermarks for speech face a step that text watermarks never do, because the detector must encode the received
audio back into tokens, and each round trip through a codec replaces some of them. We showed that these substitutions are
structured enough to design around. Our method \ours builds a basis from them and solves for embedding and detection functions that
keep the signal surviving retokenization, without changing the language model or its codec. At no more than the cost of
KGW, it keeps 80.7\,\% TPR on Moshi after eight Mimi passes, against 8.3\,\% for KGW and at most 7.3\,\% for WMAR, and it
transfers to CosyVoice3 and MOSS-TTS. Its weaknesses, frame misalignment and some foreign codecs, point to better
synchronization in the detector and a basis built from several codecs.

%% file: sections/statements.tex
\subsection*{AI use statement}
We used coding agents built on large language models as assistants. In research ideation and execution, they helped
refine the experiments we designed, helped edit the code, and helped analyze the results. They also helped find and
summarize related work, proofread the manuscript and the proofs in Appendix~\ref{app:method}, produce the figures and
tables from the evaluation outputs, and format the references. We did not use them to generate, clean or reformat data
sets. The authors reviewed all AI-assisted work, checked the code, the derivations and every number against the
evaluation outputs, and take full responsibility for the content of this paper.

\subsection*{Ethics statement}
Our aim is to make speech generated by language models traceable, which helps to identify synthetic speech used for
impersonation or fraud. A detector can also wrongly flag unwatermarked speech, so we set every threshold on unwatermarked
answers and report the false-positive rate after the attacks (Appendix~\ref{app:fpr}). A detection alone should not be
taken as proof of origin. We also report the attacks that remove our watermark, namely sub-frame shifts, changes of speed
and cropping (Appendix~\ref{app:dsp}). The TTS models read texts in the voices of speakers from the public WildVoice data
set, and the generated speech is used only for evaluation and the demo. No human subjects took part, and all data sets,
models, codecs and baselines are public and used under their licenses.

\subsection*{Reproducibility statement}
Section~\ref{sec:method} and Appendix~\ref{app:method} define the basis, the functions and the solver, with all
derivations. Section~\ref{sec:experiments} and Appendix~\ref{app:protocol} give the data sets, thresholds, choice of
strength, attacks and baseline settings. All models, codecs, baselines and data sets are public. A demo with
audio samples is at \url{https://hwiora.github.io/redwing-demo/}. Following our company's regulations,
the full code will be released after legal review.

%% file: sections/appendix.tex
\section{Method details and derivations}
\label{app:method}
This appendix follows the order of Section~\ref{sec:method}. Section~\ref{app:notation} collects the notation.
Section~\ref{app:counts} describes how the substitution counts are collected, and Section~\ref{app:basis} proves the
properties of the spectral basis. Section~\ref{app:firstorder} derives the first-order expressions for signal, cost, and
noise, and Section~\ref{app:objective} derives the objective $J$ and its maximum. Section~\ref{app:bounded} explains why the
bounds are needed and how the bounded problem is solved, Section~\ref{app:hyper} reports the choice of $K$ and $\kappa$, and
Section~\ref{app:detect} analyzes the detection statistic. Unless stated otherwise, numbers refer to Moshi.

\subsection{Notation}
\label{app:notation}
Everything in this appendix concerns one stream, whose index is omitted.

\begin{center}
\small
\begin{tabular}{p{0.2\linewidth}p{0.72\linewidth}}
\toprule
Symbol & Meaning\\
\midrule
\multicolumn{2}{l}{\textit{Tokens and retokenization}}\\
$V$ & vocabulary size of the stream\\
$X_t$, $Y_t$ & token generated at frame $t$, and the token recovered from the audio by retokenization\\
$\mathcal S$, $d_s$ & set of watermarked streams, and the generation delay of stream $s$\\
$N_{ij}$ & number of aligned frames in which token $i$ was recovered as token $j$ (Appendix~\ref{app:counts})\\
$P\in\mathbb R^{V\times V}$ & transition matrix of the retokenization channel, $P_{ij}=\Pr(Y=j\mid X=i)$, rows summing to one\\
\addlinespace
\multicolumn{2}{l}{\textit{Substitution graph and basis}}\\
$W$, $D$, $L$ & weights of the substitution graph (Eq.~\eqref{eq:graph}, without self-transitions), its degrees, and its
normalized Laplacian\\
$\lambda_k$, $K$ & $k$-th nontrivial eigenvalue of $L$, and number of basis functions ($K=16$)\\
$\phi_k\in\mathbb R^V$ & the $k$-th basis function, a real value per token\\
$\phi(i)\in\mathbb R^K$ & the values of all basis functions on token $i$\\
$\Phi\in\mathbb R^{V\times K}$ & the basis matrix, whose $k$-th column is $\phi_k$ and whose $i$-th row is $\phi(i)^\top$\\
\addlinespace
\multicolumn{2}{l}{\textit{Watermark}}\\
$b_{s,t}$ & keyed pseudorandom sign of stream $s$ at frame $t$\\
$g$, $h\in\mathbb R^V$ & embedding and detection functions, $g=\Phi a$ and $h=(\Phi-\mathbf 1\mu_0^\top)c$\\
$a$, $c\in\mathbb R^K$ & coefficients of the embedding and detection functions in the basis\\
$T$, $\delta$, $\beta$ & sampling temperature, watermark strength, and $\beta=\delta/T$\\
$\kappa$ & bound on $|g(i)|$ and $|h(i)|$ for every token ($\kappa=5$)\\
\addlinespace
\multicolumn{2}{l}{\textit{Statistics}}\\
$p_t$, $q_t$ & sampling distribution of the language model at frame $t$, without and with the bias\\
$H_0$ & unwatermarked speech from the same language model, after retokenization\\
$\pi_0\in\mathbb R^V$ & token frequencies under $H_0$\\
$A$, $B$ & surviving signal and write cost of the basis, Eqs.~\eqref{eq:A} and~\eqref{eq:B}\\
$\mu_0$, $C$ & mean $\E_{H_0}[\phi(Y)]$ and covariance $\Cov_{H_0}(\phi(Y))$ of the basis under $H_0$\\
$J$, $M$ & objective of Eq.~\eqref{eq:ratio}, and $M=B^{-1/2}AC^{-1/2}$\\
\addlinespace
\multicolumn{2}{l}{\textit{Detection}}\\
$Z(\tau)$, $Z^\star$ & detection statistic at frame offset $\tau$, and its maximum over $|\tau|\le R$ (Eq.~\eqref{eq:detector})\\
$R$ & one-sided range of the offset search ($R=2$)\\
$n$ & number of frames of a clip on one stream\\
$\mathbf 1$ & the all-ones vector\\
\bottomrule
\end{tabular}
\end{center}

\subsection{Substitution counts}
\label{app:counts}
The substitution counts record, for every token, which tokens the codec returns in its place after the audio is decoded and
encoded again. They define both the substitution graph of Eq.~\eqref{eq:graph} and the transition matrix $P$.
The counts are collected on 20 hours of LibriSpeech train-clean-100, which is disjoint from all generated speech. From the
utterances of at least one second in this 100-hour split, we draw 5,680 at random with a fixed seed until they fill
20 hours. The same utterances and the same procedure are used for all three language models.
Each utterance is
processed in four steps.
\begin{enumerate}
\item The waveform is encoded with the codec used by the language model, which gives tokens $X$, then resynthesized and
 encoded again, which gives tokens $Y$. For Moshi, resynthesis is Mimi decoding. For CosyVoice3, it is the model's
 flow-matching decoder and HiFT vocoder, prompted with the utterance itself, with the prompt trimmed to a whole number of
 token frames so that the resynthesized speech stays on the token grid. For MOSS-TTS, it is decoding with its audio
 tokenizer at all 32 quantizers.
\item The two token sequences are aligned at the frame level. Among the offsets $o\in\{-3,\ldots,3\}$, we choose the one at
 which the first stream agrees most often,
 \begin{equation}
  o^\star=\arg\max_{|o|\le3}\ \#\{t: X_{1,t}=Y_{1,t+o}\},
 \end{equation}
 and pair frame $t$ of $X$ with frame $t+o^\star$ of $Y$ on every stream. On every model, $o^\star=0$ for every utterance.
\item Two frames at each end of the utterance are dropped.
\item For each stream, $N_{ij}$ counts the aligned frames with $X_t=i$ and $Y_{t+o^\star}=j$.
\end{enumerate}

\paragraph{Support.} The support is the set of tokens with at least 50 counted occurrences as a source,
$\sum_jN_{ij}\ge50$, so that every token in the graph has enough counts to estimate its substitutions. The substitution
graph and the basis are built on the support. A token of the support that is never replaced ($D_{ii}=0$) has no edge in the
graph and receives zero basis values.
Outside the support, the rows of $\Phi$ are zero, so that
\begin{equation}
 g(i)=0,\qquad h(i)=-c^\top\mu_0\qquad\text{for tokens $i$ outside the support.}
\end{equation}
The bounds of Eq.~\eqref{eq:solve} are imposed on every token of the vocabulary, including these.

\paragraph{Transition matrix.} The transition matrix of Section~\ref{sec:moments} is the row-normalized count matrix. Unlike
the graph, it keeps the diagonal, because $P$ describes the whole channel, including the tokens that survive,
\begin{equation}
 P_{ij}=\frac{N_{ij}}{\sum_kN_{ik}}.
\end{equation}
A token that never occurs as a source has no counts. For such a token we set $P_{ii}=1$, so that $P$ is row-stochastic,
$P\mathbf 1=\mathbf 1$. These tokens lie outside the support and have zero basis values, so the completion does not change
$P\Phi$, which is all that $A$ uses.

\subsection{Spectral basis}
\label{app:basis}
Section~\ref{sec:method} uses three properties of the basis. Each basis function is smooth over the substitution graph
(Proposition~\ref{prop:rayleigh}), the basis is orthonormal and centered (Proposition~\ref{prop:orthonormal}), and every
combination of basis functions is smooth as well (Proposition~\ref{prop:combination}). This section proves them and then
explains why the graph leaves out self-transitions.

The substitution graph of Eq.~\eqref{eq:graph} has symmetric weights $W_{ij}\ge0$, no self-loops, and degrees $D_{ii}=\sum_jW_{ij}$. Its
normalized Laplacian is $L=I-D^{-1/2}WD^{-1/2}$. Throughout this section, the graph is restricted to the support tokens with
positive degree, so that $D^{-1/2}$ exists. Tokens with $D_{ii}=0$ receive zero basis values (Appendix~\ref{app:counts}).

\begin{proposition}[Smoothness identity, Eq.~\ref{eq:rayleigh}]
\label{prop:rayleigh}
Let $u$ be a unit eigenvector of $L$ with eigenvalue $\lambda$, and let $\phi=D^{-1/2}u$. Then
\begin{equation}
 \lambda=\frac12\sum_{i,j}W_{ij}\big(\phi(i)-\phi(j)\big)^2 .
\end{equation}
\end{proposition}
\begin{proof}
Because $Lu=\lambda u$ and $u^\top u=1$, $\lambda=u^\top Lu$. Substituting $u=D^{1/2}\phi$ gives
\begin{equation}
 \lambda=u^\top u-u^\top D^{-1/2}WD^{-1/2}u=\phi^\top D\phi-\phi^\top W\phi .
\end{equation}
Since $D_{ii}=\sum_jW_{ij}$ and $W$ is symmetric,
\begin{equation}
 \phi^\top D\phi=\sum_{i,j}W_{ij}\,\phi(i)^2=\frac12\sum_{i,j}W_{ij}\big(\phi(i)^2+\phi(j)^2\big),\qquad
 \phi^\top W\phi=\sum_{i,j}W_{ij}\,\phi(i)\phi(j).
\end{equation}
Subtracting the second expression from the first gives
$\frac12\sum_{i,j}W_{ij}\big(\phi(i)^2-2\phi(i)\phi(j)+\phi(j)^2\big)=\frac12\sum_{i,j}W_{ij}\big(\phi(i)-\phi(j)\big)^2$.
\end{proof}
In particular $\lambda\ge0$, and $\lambda=0$ exactly when $\phi$ is constant on every connected component of the graph. The
number of such eigenvalues equals the number of components \citep{luxburg2007tutorial}. We discard these functions by design. A
function that is constant on each component does not vary within any group of tokens that replace each other, so it
describes which component a token belongs to rather than how substitutions move within one. In practice we compute the eigenvectors of $D^{-1/2}WD^{-1/2}$, whose eigenvalues are $1-\lambda$, and keep the
$K$ largest after the trivial ones.

\begin{proposition}[Orthonormality and centering]
\label{prop:orthonormal}
Let $U=[u_1,\ldots,u_K]$ hold the retained unit eigenvectors and $\Phi=D^{-1/2}U$. Then
\begin{equation}
 \Phi^\top D\Phi=I\qquad\text{and}\qquad \mathbf 1^\top D\phi_k=0\quad\text{for every }k .
\end{equation}
\end{proposition}
\begin{proof}
First, $\Phi^\top D\Phi=U^\top D^{-1/2}DD^{-1/2}U=U^\top U=I$, because a real symmetric matrix such as $L$ has an orthonormal basis of
eigenvectors (the spectral theorem). Second, $L D^{1/2}\mathbf 1=D^{1/2}\mathbf 1-D^{-1/2}W\mathbf 1=D^{1/2}\mathbf 1-D^{-1/2}D\mathbf 1=0$, so
$D^{1/2}\mathbf 1$ is a discarded eigenvector with eigenvalue zero. Every retained $u_k$ is orthogonal to it, so
$\mathbf 1^\top D\phi_k=(D^{1/2}\mathbf 1)^\top u_k=0$.
\end{proof}
The basis is therefore orthonormal and centered under the degree weights.

\begin{proposition}[Smoothness of linear combinations]
\label{prop:combination}
Let $\lambda_1\le\cdots\le\lambda_K$ be the retained eigenvalues and $g=\Phi a$ for any $a\in\mathbb R^K$. Then
\begin{equation}
 \frac12\sum_{i,j}W_{ij}\big(g(i)-g(j)\big)^2\le\lambda_K\,g^\top Dg .
\end{equation}
The same holds for $h=(\Phi-\mathbf 1\mu_0^\top)c$.
\end{proposition}
\begin{proof}
As in the proof of Proposition~\ref{prop:rayleigh}, the left side equals $g^\top(D-W)g$. Because
$\Phi^\top(D-W)\Phi=U^\top LU=\operatorname{diag}(\lambda_1,\ldots,\lambda_K)$ and $\Phi^\top D\Phi=I$,
\begin{equation}
 g^\top(D-W)g=\sum_k\lambda_ka_k^2\le\lambda_K\sum_ka_k^2=\lambda_K\,g^\top Dg .
\end{equation}
For $h$, subtracting the constant $c^\top\mu_0$ does not change the differences $h(i)-h(j)$, so the left side is the same as
for $\Phi c$, and $h^\top Dh\ge(\Phi c)^\top D(\Phi c)$ by Proposition~\ref{prop:orthonormal}.
\end{proof}
Every function in the span of the basis is therefore at least as smooth, relative to its size, as the $K$-th basis function.

\paragraph{Why self-transitions are excluded.} A self-transition records that a token survived, not which token replaced
it. If the graph kept them, a token that often survives would be linked mostly to itself. On Moshi's deeper streams,
retaining self-transitions indeed produces leading modes concentrated on individual tokens (Table~\ref{tab:selfloops}),
instead of modes that vary smoothly over groups of tokens that replace each other. The following calculation explains why
eigenvalues can lie near the survival shares, and the localization itself is measured in the table. Suppose the substitution graph kept self-transitions, with $n_i=N_{ii}$. Its weights and
degrees would be
\begin{equation}
 W'=W+\operatorname{diag}(n),\qquad D'=D+\operatorname{diag}(n),
\end{equation}
and its normalized adjacency $S'=D'^{-1/2}W'D'^{-1/2}$ would have the diagonal entries
\begin{equation}
 r_i=S'_{ii}=\frac{n_i}{D_{ii}+n_i},
\end{equation}
the share of token $i$'s counted transitions in which it survives. Let $e_i$ be the indicator of token $i$. Then
\begin{equation}
 S'e_i=r_i\,e_i+\varepsilon_i,\qquad \|\varepsilon_i\|^2=\sum_{j\ne i}\frac{W_{ij}^2}{D'_{ii}D'_{jj}} .
\end{equation}
For a symmetric matrix, the distance from any number $r$ to the nearest eigenvalue is at most $\|S'e-re\|$ for every unit
vector $e$. Hence $S'$ has an eigenvalue within $\|\varepsilon_i\|$ of $r_i$. The residual $\varepsilon_i$ is small when token $i$
survives often and its substitutes are spread over many tokens, so the tokens that survive most reliably place eigenvalues
near their survival shares at the top of the spectrum. This bound alone does not imply that the corresponding eigenvectors
concentrate on single tokens, because eigenvalues that lie close together can mix their eigenvectors. Whether they do is an
empirical question, which Table~\ref{tab:selfloops} answers for Moshi's channel.

Table~\ref{tab:selfloops} shows this on Moshi's channel. We measure how widely a unit vector $u$ is spread by its
participation ratio $1/\sum_iu(i)^4$, the effective number of tokens it covers. With self-transitions, 11--15 of the 16
leading modes on each of streams 3--8 are single-token indicators, and their eigenvalues are within 0.07--0.14 of the
survival probability of their token. Without self-transitions, the leading modes spread over 70--230 tokens. The effect
depends on the codec. On CosyVoice3 and on MOSS-TTS heads 1--16, self-transitions narrow the modes (for CosyVoice3, the median
participation ratio falls from 264 to 110) but create almost no single-token modes (none on CosyVoice3, and one, on head
13, on MOSS-TTS). We exclude self-transitions on every model. On Moshi's deeper streams this is decisive, and elsewhere it keeps
the basis a description of substitutions rather than of survival.
\begin{table}[ht]
\caption{\textbf{Self-transitions and the basis on Moshi.} The counts come from Moshi's own Mimi channel, on the support of
tokens with at least 50 occurrences. The survival share is the fraction of counted transitions that are self-transitions.
Participation is the median participation ratio, the effective number of tokens covered, of the 16 leading nontrivial
modes of the normalized adjacency, computed without and with self-transitions. The last row counts how many of these 16
modes put at least half of their mass on one token when self-transitions are kept.}
\label{tab:selfloops}
\centering\small
\begin{tabular}{lcccccccc}
\toprule
Stream & 1 & 2 & 3 & 4 & 5 & 6 & 7 & 8\\
\midrule
Survival share & 0.72 & 0.39 & 0.26 & 0.25 & 0.20 & 0.21 & 0.19 & 0.17\\
Participation, without & 116 & 161 & 228 & 172 & 99 & 183 & 70 & 183\\
Participation, with & 79 & 151 & 1.2 & 1.1 & 1.1 & 1.2 & 1.2 & 1.0\\
Single-token modes, with & 0 & 2 & 15 & 14 & 15 & 11 & 14 & 15\\
\bottomrule
\end{tabular}
\end{table}

\subsection{First-order signal, cost, and noise}
\label{app:firstorder}
\label{app:derivation}
This section proves the three expressions of Eq.~\eqref{eq:signal}. Throughout, ``first-order'' means the leading order
in $\beta$, which is linear for the signal and quadratic for the KL cost. We first write the covariances of
Section~\ref{sec:moments} in matrix form, then state the assumptions, and then treat signal, cost, and noise in turn.

\begin{lemma}[Covariances in matrix form]
\label{lem:matrix}
For a probability vector $p$, let $\Sigma_p=\operatorname{diag}(p)-pp^\top$. For functions $f,f'\in\mathbb R^V$,
\begin{equation}
 \Cov_{X\sim p}\big(f(X),f'(X)\big)=f^\top\Sigma_pf' .
\end{equation}
Moreover $\Sigma_p\mathbf 1=0$ and $P\mathbf 1=\mathbf 1$. Consequently,
\begin{equation}
 A=\E_t\big[\Phi^\top\Sigma_{p_t}P\Phi\big],\qquad B=\E_t\big[\Phi^\top\Sigma_{p_t}\Phi\big],\qquad
 C=\Phi^\top\Sigma_{\pi_0}\Phi,\qquad \mu_0=\Phi^\top\pi_0 .
\end{equation}
\end{lemma}
\begin{proof}
The covariance is $\sum_ip(i)f(i)f'(i)-\big(\sum_ip(i)f(i)\big)\big(\sum_ip(i)f'(i)\big)=f^\top\operatorname{diag}(p)f'-f^\top pp^\top f'$.
Next, $\Sigma_p\mathbf 1=p-p\,(p^\top\mathbf 1)=p-p=0$, and $P\mathbf 1=\mathbf 1$ because each row of $P$ sums to one. The matrix
forms follow by applying the first identity to each pair of basis functions, with $f=\phi_k$ and $f'=\phi_l$ for $B$ and $C$,
and $f=\phi_k$ and $f'=P\phi_l$ for $A$. Under $H_0$, the tokens are distributed as $\pi_0$, which gives $C$ and
$\mu_0=\sum_i\pi_0(i)\phi(i)=\Phi^\top\pi_0$.
\end{proof}

\paragraph{Assumptions.} Fix one frame, and let $p$ be the sampling distribution of the language model at that frame. We
assume the following.
\begin{enumerate}
\item The sampling support is fixed (for example, the top-$k$ set), and on it $p(i)\propto\exp(\ell(i)/T)$.
\item The sign $b$ of the current frame is uniform on $\{-1,+1\}$ and independent of the context that determines $p$.
\item Given the generated token $X$, the recovered token $Y$ follows the transition matrix $P$, independently of $b$.
\end{enumerate}
The first assumption says that the sampler is a softmax at temperature $T$ over a set of candidate tokens that the bias
does not change. The second holds because the signs
come from a secret key. The third treats retokenization as a channel that acts on each token without knowledge of the key.
The remark at the end of this section lists what these assumptions leave out.
Under the first assumption, the biased logits of Eq.~\eqref{eq:write} give the tilted distribution
\begin{equation}
 q_\beta(i)=p(i)\,e^{\beta bg(i)-\psi(\beta)},\qquad \psi(\beta)=\log\sum_ip(i)\,e^{\beta bg(i)} .
\label{eq:tilt}
\end{equation}

\begin{lemma}[Derivative of a tilted expectation]
\label{lem:derivative}
For every function $f\in\mathbb R^V$,
\begin{equation}
 \frac{d}{d\beta}\E_{q_\beta}\big[f(X)\big]=b\,\Cov_{q_\beta}\big(g(X),f(X)\big).
\end{equation}
\end{lemma}
\begin{proof}
Differentiating $\E_{q_\beta}[f]=\sum_if(i)p(i)e^{\beta bg(i)-\psi(\beta)}$ term by term gives
\begin{equation}
 \frac{d}{d\beta}\E_{q_\beta}[f]=\sum_if(i)\,q_\beta(i)\big(bg(i)-\psi'(\beta)\big).
\end{equation}
Differentiating $\psi$ gives $\psi'(\beta)=b\sum_iq_\beta(i)g(i)=b\,\E_{q_\beta}[g]$. Substituting,
$\frac{d}{d\beta}\E_{q_\beta}[f]=b\big(\E_{q_\beta}[fg]-\E_{q_\beta}[f]\,\E_{q_\beta}[g]\big)=b\,\Cov_{q_\beta}(g,f)$.
\end{proof}

\begin{proposition}[Signal]
\label{prop:signal}
Averaged over the sign $b$,
\begin{equation}
 \E\big[b\,h(Y)\big]=\beta\,a^\top\Phi^\top\Sigma_pP\Phi\,c+O(\beta^2),
\end{equation}
and averaging over frames gives $\beta\,a^\top Ac+O(\beta^2)$ with $A$ as in Eq.~\eqref{eq:A}.
\end{proposition}
\begin{proof}
By the third assumption, $\E[h(Y)\mid X]=(Ph)(X)$, so $\E[h(Y)\mid b]=\E_{q_\beta}[(Ph)(X)]$. A first-order expansion around
$\beta=0$ with Lemma~\ref{lem:derivative} gives
\begin{equation}
 \E_{q_\beta}\big[(Ph)(X)\big]=\E_p\big[(Ph)(X)\big]+\beta\,b\,\Cov_p\big(g(X),(Ph)(X)\big)+O(\beta^2).
\end{equation}
Multiplying by $b$ and using $b^2=1$,
\begin{equation}
 \E\big[b\,h(Y)\mid b\big]=b\,\E_p\big[(Ph)(X)\big]+\beta\,\Cov_p\big(g(X),(Ph)(X)\big)+O(\beta^2).
\end{equation}
By the second assumption, $b$ is independent of $p$ and has mean zero, so averaging over $b$ removes the first term. By
Lemma~\ref{lem:matrix}, the covariance equals $g^\top\Sigma_pPh$. Substituting $g=\Phi a$ and $h=\Phi c-\mathbf 1\mu_0^\top c$,
\begin{equation}
 g^\top\Sigma_pPh=a^\top\Phi^\top\Sigma_pP\Phi c-a^\top\Phi^\top\Sigma_pP\mathbf 1\,(\mu_0^\top c)=a^\top\Phi^\top\Sigma_pP\Phi c,
\end{equation}
because $\Sigma_pP\mathbf 1=\Sigma_p\mathbf 1=0$.
\end{proof}

\begin{proposition}[Cost]
\label{prop:cost}
The KL divergence of the tilted distribution from the original one is
\begin{equation}
 \KL(q_\beta\Vert p)=\tfrac12\beta^2\,a^\top\Phi^\top\Sigma_p\Phi\,a+O(\beta^3),
\end{equation}
and averaging over frames gives $\tfrac12\beta^2\,a^\top Ba+O(\beta^3)$ with $B$ as in Eq.~\eqref{eq:B}.
\end{proposition}
\begin{proof}
From Eq.~\eqref{eq:tilt}, $\log\big(q_\beta(i)/p(i)\big)=\beta bg(i)-\psi(\beta)$, so
\begin{equation}
 \KL(q_\beta\Vert p)=\E_{q_\beta}\big[\beta bg(X)-\psi(\beta)\big]=\beta\psi'(\beta)-\psi(\beta),
\end{equation}
using $\psi'(\beta)=b\,\E_{q_\beta}[g]$ from the proof of Lemma~\ref{lem:derivative}. At $\beta=0$, $\psi(0)=0$,
$\psi'(0)=b\,\E_p[g]$, and $\psi''(0)=b^2\Var_p(g)=\Var_p(g)$. Taylor expansion of both terms gives
\begin{equation}
 \beta\big(\psi'(0)+\beta\psi''(0)\big)-\Big(\beta\psi'(0)+\tfrac12\beta^2\psi''(0)\Big)+O(\beta^3)=\tfrac12\beta^2\Var_p(g)+O(\beta^3).
\end{equation}
Finally, $\Var_p(g)=g^\top\Sigma_pg=a^\top\Phi^\top\Sigma_p\Phi a$ by Lemma~\ref{lem:matrix}.
\end{proof}

\begin{proposition}[Noise, and the effect of centering]
\label{prop:noise}
On unwatermarked speech, one frame's contribution $b\,h(Y)$ has mean zero and variance
\begin{equation}
 \E_{H_0}\big[h(Y)^2\big]=c^\top Cc .
\end{equation}
For the uncentered function $\tilde h=\Phi c$, the variance is $c^\top Cc+(c^\top\mu_0)^2$, while the signal of
Proposition~\ref{prop:signal} is the same as for $h$.
\end{proposition}
\begin{proof}
Under $H_0$ the sign $b$ is independent of $Y$ and has mean zero, so $\E[b\,h(Y)]=0$ and
$\E[(b\,h(Y))^2]=\E_{H_0}[h(Y)^2]$. For any function $f$, $\E_{H_0}[f(Y)^2]=\Var_{H_0}(f(Y))+(\pi_0^\top f)^2$. For the centered
function,
\begin{equation}
 \pi_0^\top h=\pi_0^\top\Phi c-(\pi_0^\top\mathbf 1)\,\mu_0^\top c=\mu_0^\top c-\mu_0^\top c=0,
\end{equation}
and $\Var_{H_0}(h(Y))=h^\top\Sigma_{\pi_0}h=c^\top\Phi^\top\Sigma_{\pi_0}\Phi c=c^\top Cc$, because $\Sigma_{\pi_0}\mathbf 1=0$. For
$\tilde h=\Phi c$, the variance is the same and the mean is $\pi_0^\top\Phi c=c^\top\mu_0$, which gives the second statement. The
signal is unchanged because the covariance in Proposition~\ref{prop:signal} does not depend on constant shifts of $h$.
\end{proof}
Centering therefore removes the variance $(c^\top\mu_0)^2$ per frame without changing the signal. In the detector this variance
appears as the term $m\sum_tb_t$ with $m=c^\top\mu_0$. The embedding function needs no centering, because adding the same
constant to every logit leaves the sampling distribution unchanged.

\begin{remark}[Scope]
These are leading-order expansions at $\delta=0$ for a fixed sampling support, linear in $\beta$ for the signal and
quadratic for the KL cost. They do not describe support
changes caused by the bias, samplers with additional renormalization, or the effect of the watermark on later contexts. They
are used to design the functions, and the behavior at the deployed strength is measured in Section~\ref{sec:results}.
\end{remark}

\subsection{Objective}
\label{app:objective}
This section explains the objective $J$ of Section~\ref{sec:solve}, states the approximation behind it, and derives its
maximum.

\paragraph{Where the signal comes from.} The threshold is set on unwatermarked audio, where $Z$ has mean zero and unit
variance, so a pair of functions is better the higher it raises the mean of $Z$ on watermarked audio. To first order in
$\beta$, the bias multiplies the probability of token $i$ by $1+\beta b\,(g(i)-\bar g)$, where $\bar g$ is the mean of $g$
under $p_t$. In a frame with $b=+1$, tokens with large $g$ therefore become more likely, and in a frame with $b=-1$ less
likely. Retokenization carries part of this change to the recovered token through $P$, and the score $h$ picks it up.
Averaged over the sign, the product $b\,h(Y)$ is thus $\beta$ times the covariance between $g$ of the generated token and
$h$ of the recovered one, which in the basis is $a^\top Ac$ (Proposition~\ref{prop:signal}).

\paragraph{From one frame to a clip.} A clip has $n$ frames. The numerator of Eq.~\eqref{eq:detector} is a sum of $n$
terms $b\,h$, each with mean $\beta\,a^\top Ac$, so it is about $n\beta\,a^\top Ac$. The denominator is the square root of a
sum of $n$ squared scores, each about $c^\top Cc$ (Proposition~\ref{prop:noise}), so it is about $\sqrt{n\,c^\top Cc}$. The
mean of $Z$ is their ratio, $\sqrt n\,\beta\,a^\top Ac/\sqrt{c^\top Cc}$, and grows with the square root of the clip
length.

\paragraph{At a fixed cost.} We compare pairs at the same cost. For a budget of $\varepsilon$ nats per frame,
$\tfrac12\beta^2a^\top Ba=\varepsilon$ (Proposition~\ref{prop:cost}) fixes the strength at $\beta=\sqrt{2\varepsilon/a^\top
Ba}$, and substituting gives $\E[Z]\approx\sqrt{2\varepsilon n}\,J(a,c)$. The first factor depends only on the budget and
the clip length, and $J$ only on the pair. The next paragraph states the approximation behind these steps, and
Proposition~\ref{prop:svd} gives the maximum of $J$.

\paragraph{The approximation.} The steps above replace the mean of the ratio $Z$ by the ratio of the mean numerator to the
typical denominator. This is an approximation, not a derived result. It treats $\frac1n\sum_th(Y_t)^2$ as close to its
mean $c^\top Cc+O(\beta)$ (Proposition~\ref{prop:noise}), although the tokens of a clip are generated autoregressively and
are not independent, and it keeps only the first order in $\beta$. Under it, for one stream with $n$ frames at the correct
offset,
\begin{equation}
 \E[Z]\approx\sqrt n\,\beta\,\frac{a^\top Ac}{\sqrt{c^\top Cc}}=\sqrt{2\varepsilon n}\,J(a,c)
 \quad\text{at a first-order cost }\varepsilon=\tfrac12\beta^2a^\top Ba\text{ per frame.}
\label{eq:expected}
\end{equation}
Under this approximation, maximizing $J$ maximizes the predicted mean shift of $Z$ at a given cost and clip length. $J$ is
therefore a design objective. It does not guarantee an ordering of the TPR, which depends on the whole distribution of $Z$,
and the TPR is measured directly (Section~\ref{sec:results}).

\begin{proposition}[Maximum of $J$]
\label{prop:svd}
Suppose $B$ and $C$ are positive definite, and let $M=B^{-1/2}AC^{-1/2}$ with largest singular value $\sigma_1(M)$ and leading
left and right singular vectors $\tilde u_1,\tilde v_1$. Then
\begin{equation}
 \max_{a,c\ne0}J(a,c)=\sigma_1(M),
\end{equation}
attained at $a=B^{-1/2}\tilde u_1$ and $c=C^{-1/2}\tilde v_1$ (and their positive multiples). The same value and maximizers solve
$\max\{a^\top Ac:\ a^\top Ba\le1,\ c^\top Cc\le1\}$.
\end{proposition}
\begin{proof}
Substitute $\tilde a=B^{1/2}a$ and $\tilde c=C^{1/2}c$. Then $a^\top Ba=\|\tilde a\|^2$, $c^\top Cc=\|\tilde c\|^2$, and
$a^\top Ac=\tilde a^\top M\tilde c$, so
\begin{equation}
 J(a,c)=\frac{\tilde a^\top M\tilde c}{\|\tilde a\|\,\|\tilde c\|}\le\frac{\|\tilde a\|\,\|M\tilde c\|}{\|\tilde a\|\,\|\tilde c\|}\le\sigma_1(M).
\end{equation}
The first inequality is the Cauchy--Schwarz inequality and the second is the definition of $\sigma_1(M)$ as the largest value
of $\|M\tilde c\|/\|\tilde c\|$. Both are equalities at $\tilde a=\tilde u_1$ and $\tilde c=\tilde v_1$. For the constrained form,
$a^\top Ac=\tilde a^\top M\tilde c\le\|\tilde a\|\,\|\tilde c\|\,\sigma_1(M)\le\sigma_1(M)$ on the feasible set, with equality at
the same point.
\end{proof}
If $B$ or $C$ is singular, the same argument applies after restricting $a$ and $c$ to the ranges of $B$ and $C$.

\begin{corollary}[Best detection function for a fixed embedding function]
\label{cor:matched}
Suppose $C$ is positive definite and $A^\top a\ne0$. For a fixed $a$, $J(a,c)$ is maximized over $c\ne0$ by the positive
multiples of $c^\star=C^{-1}A^\top a$.
\end{corollary}
\begin{proof}
With $\tilde c=C^{1/2}c$, we have $a^\top Ac=(C^{-1/2}A^\top a)^\top\tilde c$ and $c^\top Cc=\|\tilde c\|^2$, so $J(a,c)$ is
proportional to $(C^{-1/2}A^\top a)^\top\tilde c/\|\tilde c\|$. By the Cauchy--Schwarz inequality, this is largest when
$\tilde c$ is a positive multiple of $C^{-1/2}A^\top a$, that is, when $c$ is a positive multiple of $C^{-1}A^\top a$.
\end{proof}
The detection coefficients are thus a matched filter \citep{turin1960matched}. $A^\top a$ is the embedded signal as it arrives after
retokenization, and $C^{-1}$ discounts the directions in which the detection function is noisy on unwatermarked speech.

\subsection{Bounded problem and solver}
\label{app:bounded}
\paragraph{Why pointwise bounds are needed.} The cost $a^\top Ba$ averages $\Var_{p_t}(g)$ over frames. A token that the
language model rarely samples contributes almost nothing to it, so it can take a very large value at almost no cost. The
maximizer of $J$ uses this freedom. On Moshi, its largest values on streams 1--4 are 23, 217, 50, and 5.5 standard
deviations. When such a token enters the candidates of the language model, Eq.~\eqref{eq:write} multiplies its probability by
$e^{\pm\beta g(i)}$. At $\delta=0.07$ and $T=0.8$, for example, a value of 217 gives
\begin{equation}
 \beta\,g(i)=\frac{0.07}{0.8}\times217\approx19,\qquad e^{\pm19},
\end{equation}
which forces or excludes the token. The first-order cost of Proposition~\ref{prop:cost} no longer describes such a change.
Table~\ref{tab:uncapped} compares the unbounded and the bounded functions on the validation set. The unbounded functions
already lose 0.48 of predicted naturalness (UTMOSv2) at $\delta=0.07$, where they detect only 35\,\% of the clips without any
attack. The bounded functions at $\delta=0.7$ spend more KL divergence, keep predicted naturalness, and detect 97\,\% of the clips.
On the detection side, a token with an extreme value of $h$ would dominate the sum in Eq.~\eqref{eq:detector}, so that the
decision would depend on whether that token occurs rather than on evidence spread over many frames.
\begin{table}[ht]
\caption{\textbf{Unbounded and bounded functions.} All values are for Moshi's streams 1--4 on the validation set (250
prompts). KL is the divergence per frame, summed over the four streams. The UTMOSv2 difference is paired against
unwatermarked speech, and the TPR is at the fixed threshold without attack.}
\label{tab:uncapped}
\centering\small
\begin{tabular}{lcccc}
\toprule
Functions & $\delta$ & KL (nats/frame) & UTMOSv2 difference & TPR (\%)\\
\midrule
Unbounded & 0.05 & 0.05 & $-0.04$ & 2\\
Unbounded & 0.06 & 0.24 & $-0.05$ & 7\\
Unbounded & 0.07 & 0.94 & $-0.48$ & 35\\
Unbounded & 0.08 & 1.94 & $-1.00$ & 68\\
Bounded ($\kappa=5$) & 0.70 & 1.54 & $+0.03$ & 97\\
\bottomrule
\end{tabular}
\end{table}

\paragraph{Why the bounds are placed on the scale-fixed problem.} $J$ is unchanged when $a$ or $c$ is multiplied by a positive
constant. A bound such as $|g(i)|\le\kappa$ added to the maximization of $J$ would therefore constrain nothing, because any
direction can be shrunk until it satisfies the bound, and shrinking does not change $J$. Eq.~\eqref{eq:solve} instead first
fixes the scale with $a^\top Ba\le1$ and $c^\top Cc\le1$, and at this scale the bound has a meaning. At $a^\top Ba=1$, the
frame-averaged variance of $g$ under the sampling distribution is one, so the typical logit shift is $\delta$, and
$|g(i)|\le\kappa$ bounds each token's shift by $\kappa\delta$, that is, by $\kappa$ times the typical shift. Likewise, $c^\top Cc=1$ gives $h$ unit variance per frame under $H_0$.
By Proposition~\ref{prop:cost}, the first-order cost for the tempered softmax sampler is then at most $\tfrac12(\delta/T)^2$ nats
per stream and frame.

\paragraph{What the bounded problem optimizes.} Without the amplitude bounds, Eq.~\eqref{eq:solve} has the same maximizers as
$J$ (Proposition~\ref{prop:svd}). With them, it is a surrogate. It maximizes the signal $a^\top Ac$ under upper bounds on the
cost and the noise, and either quadratic constraint can remain slack, in which case $a^\top Ac$ and $J$ differ. On Moshi, the
cost constraint is slack on stream 1 ($a_1^\top B_1a_1=0.92$) and active on streams 2--4 (Appendix~\ref{app:strength}). The
bounded solution therefore does not inherit the optimality of Proposition~\ref{prop:svd}, and Table~\ref{tab:kappa} measures
how much of $\sigma_1(M)$ it retains.

\paragraph{Solver.} For fixed $c$, Eq.~\eqref{eq:solve} reduces to
\begin{equation}
 \max_a\ (Ac)^\top a\quad\text{subject to}\quad \|B^{1/2}a\|\le1,\quad |(\Phi a)_i|\le\kappa\ \text{for all } i,
\end{equation}
a linear objective over the intersection of a second-order cone and a polytope. For fixed $a$, the problem in $c$ has the same
form, with $A^\top a$, $C^{1/2}$, and $\Phi-\mathbf 1\mu_0^\top$ in place of $Ac$, $B^{1/2}$, and $\Phi$. The procedure is as follows.
\begin{enumerate}
\item Start from one of the three leading singular vector pairs of $M$ (Proposition~\ref{prop:svd}).
\item Solve the program in $a$ with $c$ fixed, then the program in $c$ with $a$ fixed. Both are solved with CVXPY
 \citep{diamond2016cvxpy} and Clarabel \citep{goulart2024clarabel}, with absolute and relative gap and feasibility tolerances of
 $10^{-9}$.
\item Repeat step 2 until the objective changes by less than a relative $10^{-7}$, or for at most 150 iterations.
\item Keep the best result over the three starting points.
\end{enumerate}

\begin{proposition}[Monotone ascent]
\label{prop:ascent}
Suppose each block step is solved exactly. From the first complete sweep on, the objective values $a^\top Ac$ produced by
the procedure are nondecreasing and converge.
\end{proposition}
\begin{proof}
The singular vectors used as starting points need not satisfy the amplitude bounds, but after the first complete sweep both
blocks are feasible. From then on, each step maximizes the objective over one block of variables with the other fixed. The
current value of that block is feasible for the step, so the objective cannot decrease. The feasible set is contained in $\{a^\top Ba\le1,\ c^\top Cc\le1\}$,
where the objective is at most $\sigma_1(M)$ by Proposition~\ref{prop:svd}. A nondecreasing bounded sequence converges.
\end{proof}
With a numerical solver, these statements hold up to the solver's tolerance. The joint problem is not convex, so the limit
need not be the global maximum. At $\kappa=5$, the attained objective is 68--100\,\%
of $\sigma_1(M)$ per stream on Moshi (Table~\ref{tab:kappa}).

\subsection{Choice of $K$ and $\kappa$}
\label{app:hyper}
\paragraph{Basis dimension $K$.} For each stream and each $k=1,\ldots,16$, we solve Eq.~\eqref{eq:solve} with $\kappa=5$ inside the
first $k$ basis functions. To see whether a larger $k$ only fits noise, we split the fitting data in half. $A$ and $B$ are
computed from the sampling distributions recorded during generation, and $C$ from the tokens that the detector recovers
from the resulting audio, so the two kinds of data are split separately, the sampling distributions by generation shard
and the recovered tokens by clip. We solve on one half, evaluate the objective of the solution on the other half after
renormalizing it to $a^\top Ba=c^\top Cc=1$, and average the two directions. Figure~\ref{fig:kdim} shows the result. The held-out objective is 97--100\,\% of the in-sample objective at
every $k$, so a larger $K$ is not penalized by estimation noise. Streams 5--8 saturate by $k=6$ and streams 1--4 by $k=12$, and
every stream is within 1\,\% of its value at $k=16$ already at $k=15$. $K=16$ is thus the smallest power of two within 5\,\% of
the plateau on every stream. The objective does not use the modes in the order of their eigenvalues. The most useful single
mode is mode 9 on stream 1 and mode 14 on stream 2, so a cutoff on the eigenvalue would discard useful modes.
\begin{figure}[ht]
\centering
\includegraphics[width=\linewidth]{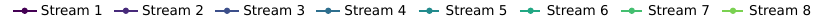}\\[2pt]
\begin{subfigure}[t]{0.418\linewidth}
\centering
\includegraphics[width=\linewidth]{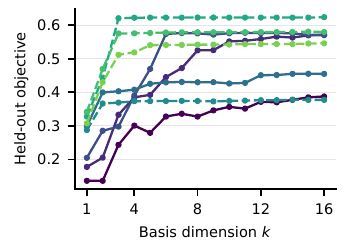}
\caption{Held-out objective}
\label{fig:kdim_a}
\end{subfigure}\hspace{0.03\linewidth}
\begin{subfigure}[t]{0.418\linewidth}
\centering
\includegraphics[width=\linewidth]{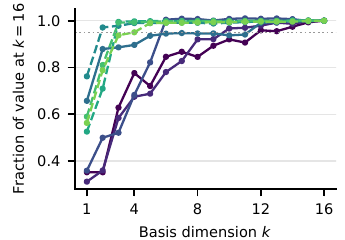}
\caption{Relative to $k=16$}
\label{fig:kdim_b}
\end{subfigure}
\caption{\textbf{Held-out objective against the basis dimension $k$.} The objective is estimated per Moshi stream with
$\kappa=5$, fitting on one half of the data and evaluating on the other. Solid lines are the watermarked streams 1--4. The dotted line in (b) marks 95\,\% of the value at $k=16$.}
\label{fig:kdim}
\end{figure}

\paragraph{Bound $\kappa$.} Table~\ref{tab:kappa} gives the objective of the bounded problem as a fraction of the unconstrained
optimum $\sigma_1(M)$. The mean over streams rises from 0.76 at $\kappa=3$ to 0.87 at $\kappa=5$, and only to 0.91 at $\kappa=10$.
Among the tested values, $\kappa=5$ is the smallest within 5\,\% of the value at $\kappa=10$. We also generated speech on 100
validation prompts for $\kappa\in\{3,4,4.5,5,5.5,6,7,10\}$ and $\delta\in\{0.6,0.7,0.9,1.0\}$. Detection per unit of KL divergence
is flat for $\kappa$ between 4.5 and 6, and quality drops sharply at $\kappa=10$ (UTMOSv2 difference of $-1.0$ at $\delta=0.9$).
\begin{table}[ht]
\caption{\textbf{Objective retained under the bound $\kappa$.} Each value is the attained objective divided by the
unconstrained optimum $\sigma_1(M)$, for one Moshi stream.}
\label{tab:kappa}
\centering\small
\begin{tabular}{lccccccccc}
\toprule
& \multicolumn{8}{c}{Stream} & \\
\cmidrule(lr){2-9}
$\kappa$ & 1 & 2 & 3 & 4 & 5 & 6 & 7 & 8 & Mean\\
\midrule
2 & 0.31 & 0.26 & 0.41 & 0.70 & 0.57 & 0.58 & 0.64 & 0.63 & 0.51\\
3 & 0.49 & 0.49 & 0.79 & 0.94 & 0.78 & 0.86 & 0.84 & 0.85 & 0.76\\
4 & 0.65 & 0.65 & 0.89 & 0.98 & 0.79 & 0.96 & 0.85 & 0.90 & 0.83\\
5 & 0.82 & 0.68 & 0.90 & 1.00 & 0.79 & 1.00 & 0.87 & 0.93 & 0.87\\
7 & 0.91 & 0.69 & 0.90 & 1.00 & 0.80 & 1.00 & 0.90 & 0.96 & 0.90\\
10 & 0.95 & 0.70 & 0.90 & 1.00 & 0.80 & 1.00 & 0.94 & 1.00 & 0.91\\
\bottomrule
\end{tabular}
\end{table}

\subsection{Detection statistic}
\label{app:detect}
\paragraph{Stream delays.} Some language models generate their streams with fixed delays, so that the token of stream $s$ at
output frame $t$ is generated at step $t+d_s$. Moshi uses $d_1=0$ and $d_2=\cdots=d_8=1$. MOSS-TTS uses $d_s=s-1$. CosyVoice3 has a
single stream with $d=0$. The detector indexes the key signs accordingly.

\paragraph{Null distribution.} The results below idealize the keyed pseudorandom signs as independent fair signs that are
independent of the tokens of unwatermarked audio. Probabilities are then over the signs, with the clip fixed. For one fixed,
deployed key this is a model rather than a guarantee, since the secrecy of the key alone does not make the signs
statistically independent of every clip. We also assume that the clip has at least one frame with $h_s(Y_{s,t})\ne0$, so
that the denominator of Eq.~\eqref{eq:detector} is positive. The empirical false-positive rates in
Appendix~\ref{app:fpr} check the model on real unwatermarked audio.

\begin{proposition}[Null moments]
\label{prop:nullmoments}
On unwatermarked audio, fix the recovered tokens and an offset $\tau$. Then $Z(\tau)$ has conditional mean zero and conditional
variance one.
\end{proposition}
\begin{proof}
Write $w_{s,t}=h_s(Y_{s,t})$ and $\sigma^2=\sum_{s,t}w_{s,t}^2$. The numerator of Eq.~\eqref{eq:detector} is
$\sum_{s,t}b_{s,t+\tau+d_s}w_{s,t}$. For a fixed $\tau$, distinct pairs $(s,t)$ use distinct signs, and on unwatermarked audio the
signs are independent of the tokens. The numerator is therefore a sum of independent terms $\pm w_{s,t}$ with mean zero and
variance $w_{s,t}^2$. Its conditional mean is zero and its conditional variance is $\sigma^2$, the square of the denominator.
\end{proof}

\begin{corollary}[Tail bound]
\label{cor:tail}
For every $z>0$ and every clip,
\begin{equation}
 \Pr\big(Z(\tau)>z\big)\le e^{-z^2/2},\qquad \Pr\big(Z^\star>z\big)\le(2R+1)\,e^{-z^2/2},
\end{equation}
where $R$ is the one-sided range of the offset search.
\end{corollary}
\begin{proof}
Hoeffding's inequality for a sum of independent fair signs with weights $w_{s,t}$ gives
$\Pr\big(\sum b\,w>z\sigma\big)\le\exp\!\big(-z^2\sigma^2/(2\sigma^2)\big)=e^{-z^2/2}$, conditionally on the tokens and therefore
also unconditionally. The second bound follows from the union bound over the $2R+1$ offsets.
\end{proof}
The tail bound holds for every clip and does not need a central limit theorem. $Z(\tau)$ is moreover asymptotically
standard normal when the largest normalized weight vanishes, $\max_{s,t}|w_{s,t}|/\sigma\to0$. Because $|h_s|\le\kappa$,
this holds whenever $\sigma$ grows without bound. More frames alone do not guarantee it, since frames with zero score add
nothing to $\sigma$. With $R=2$, Corollary~\ref{cor:tail} gives a
conservative 1\,\% threshold of $\sqrt{2\ln500}\approx3.53$. The empirical thresholds are lower (2.94 for Moshi), and we use the
empirical ones (Section~\ref{sec:experiments}). These bounds do not depend on the length of the clip, whereas on
watermarked clips the expected statistic at the correct offset grows as $\sqrt n$ with the number of frames $n$
(Eq.~\eqref{eq:expected}).

\paragraph{Offset search.} Table~\ref{tab:bestshift} gives the offset at which $Z(\tau)$ is largest. On watermarked clips it
concentrates on one value per model. It is $\tau=-1$ on Moshi, whose own decoding and re-encoding shifts the tokens by one
frame, and $\tau=0$ on CosyVoice3 and MOSS-TTS. On unwatermarked clips it is spread evenly over the searched offsets. A window of
$\pm2$ frames therefore covers the offsets at which watermarked clips carry evidence. A wider window would only raise the
maximum on unwatermarked clips, and with it the threshold. After eight passes through Mimi, EnCodec, SpeechTokenizer, DAC44, or
Opus, the Moshi peak stays at $\tau=-1$ on 94--99\,\% of watermarked clips.
\begin{table}[ht]
\caption{\textbf{Where the evidence peaks.} Each value is the fraction of test clips whose maximizing offset $\tau^\star$
takes the given value, when the search covers $\pm8$ frames and no attack is applied. An even spread would put 1/17, or
0.059, on each offset.}
\label{tab:bestshift}
\centering\small
\begin{tabular}{llccccc}
\toprule
Model & Clips & $\tau^\star=-2$ & $-1$ & $0$ & $+1$ & $+2$\\
\midrule
Moshi & watermarked (600) & 0.000 & 0.985 & 0.000 & 0.000 & 0.000\\
Moshi & unwatermarked (600) & 0.055 & 0.073 & 0.047 & 0.058 & 0.070\\
CosyVoice3 & watermarked (596) & 0.002 & 0.000 & 0.988 & 0.002 & 0.002\\
CosyVoice3 & unwatermarked (597) & 0.055 & 0.047 & 0.062 & 0.049 & 0.060\\
MOSS-TTS & watermarked (593) & 0.000 & 0.000 & 1.000 & 0.000 & 0.000\\
MOSS-TTS & unwatermarked (599) & 0.048 & 0.047 & 0.078 & 0.053 & 0.060\\
\bottomrule
\end{tabular}
\end{table}

\FloatBarrier
\section{Experimental details}
\label{app:protocol}
This appendix describes the three language models (Appendix~\ref{app:models}), the data and thresholds
(Appendix~\ref{app:data}), the choice of the strength $\delta$ (Appendix~\ref{app:strength}), the attacks
(Appendix~\ref{app:attacks}), the baselines (Appendix~\ref{app:baselines}), and the scoring of speech quality
(Appendix~\ref{app:qualityprotocol}).

\subsection{Language models}
\label{app:models}
\paragraph{Moshi.}
We use Moshiko with fp16 inference. Audio streams are sampled at temperature 0.8 with top-$k=250$, and the text stream at
temperature 0.7 with top-$k=25$, the defaults of the released generator. The watermark bias is added to the audio logits
before temperature scaling. Generation stops after at least 38 answer frames followed by 25 consecutive padding frames on
the text stream, or at 200 frames (16\,s). Unwatermarked and watermarked answers to a prompt share the generation seed.

\paragraph{CosyVoice3.}
CosyVoice3 generates one stream of 6,561 tokens at 25 frames per second. We use Fun-CosyVoice3-0.5B with its repetition-aware sampler: top-$p=0.8$, top-$k=25$, and a token is redrawn when it
fills more than a tenth of the last ten positions. After the watermark bias, the distribution is renormalized over speech
tokens. These steps are why the first-order cost of Proposition~\ref{prop:cost} underestimates the measured cost on this
model (Section~\ref{sec:experiments}). The detector tokenizes 24\,kHz audio with the model's speech tokenizer, and the
native resynthesis decodes with the flow-matching decoder and the HiFT vocoder.

\paragraph{MOSS-TTS.}
MOSS-TTS generates 32 streams of 1,024 tokens at 12.5 frames per second. We use MOSS-TTS Delay-8B with its released sampler: temperature 1.7, top-$k=25$ and top-$p=0.8$. Stream $s$ is delayed by
$d_s=s-1$ frames. We watermark only streams 1--16, because streams 17--32 carry little of the watermark through
resynthesis. In a validation run that watermarked all 32 streams, a detector that reads streams 1--16 detected 92.3\,\% of
the clips after eight native passes, as many as a detector that reads all 32. Streams 17--32 alone detected 24.6\,\%, and
adding them to the detector lowered the TPR after foreign codecs, for example after eight EnCodec passes from 72.2 to
62.5\,\%.

\subsection{Data and thresholds}
\label{app:data}
\paragraph{Data sets.}
The evaluation protocol is the same on every language model (Table~\ref{tab:splits}). Each model has 250 validation and
600 test prompts, drawn at random from a pool of 1,000. The strength is chosen on the validation set, every
reported result comes from the test set, and the thresholds come from calibration answers that are used for nothing
else. Moshi answers spoken questions from WildVoice, and CosyVoice3 and MOSS-TTS read texts in the voices of WildVoice
speakers. On every model, the unwatermarked answers that fit the statistics of Section~\ref{sec:moments} are disjoint
from the validation, calibration and test sets. The substitution counts come from
LibriSpeech for all three models (Appendix~\ref{app:counts}).

\begin{table}[h]
\caption{\textbf{Data sets.} The numbers count prompts for Moshi and texts for the TTS models, and the calibration sets
count unwatermarked answers. On the TTS models, a few generations failed
or exceeded the tokenizer's length limit and are excluded, and parentheses give the clips left.}
\label{tab:splits}
\centering\small
\begin{tabular}{llccc}
\toprule
Set & Used for & Moshi & CosyVoice3 & MOSS-TTS\\
\midrule
Substitution counts & graph, basis and $P$ & \multicolumn{3}{c}{20 hours of LibriSpeech}\\
Validation & strength $\delta$ and watermarked streams & 250 & 250 & 250\\
Calibration & thresholds & 1,200 & 1,000 (994) & 1,000 (993)\\
Test & all reported results & 600 & 600 (596--598) & 600 (593--599)\\
\bottomrule
\end{tabular}
\end{table}

\paragraph{Calibration and thresholds.}
The thresholds are set on unwatermarked answers that are used for nothing else. For Moshi, these are 1,200 answers, three
new generations for each fitting and validation prompt. For each TTS model, they are the answers to 1,000 separate texts.
Each detector scores its calibration answers, and its threshold is the $k$-th largest score, with $k=\lceil0.01n\rceil$ for
$n$ answers. A clip is flagged when its score is strictly above the threshold, so at most $k-1$ calibration answers are
flagged, for example 11 of 1,200 (0.92\,\%) on Moshi. The same threshold is then used under every attack. WMAR changes the
decoder, so its thresholds come from unwatermarked answers decoded and re-encoded by its own fine-tuned codec
(Appendix~\ref{app:baselines}).

\paragraph{Sampling noise in the FPR.}
Without attack, the FPR on the unwatermarked test answers is close to 1\,\% by construction. Under attack, nothing ties it
to 1\,\%, so it can move in either direction. With 600 clips, an FPR of 0.5\,\% corresponds to three clips and has a
95\,\% Wilson interval of 0.2 to 1.5\,\%.

\subsection{Strength and cost}
\label{app:strength}
\paragraph{Validation grids.}
The budget is the validation cost of KGW with $\delta=2$ on the same streams, and the strength is the largest value on the
grid whose validation cost stays within it. The grids are $\{0.5, 0.6, 0.7, 0.8, 0.9\}$ on Moshi, $\{0.7, 0.9, 1.1\}$ on CosyVoice3 and $\{0.7, 0.8, \ldots, 1.2\}$ on MOSS-TTS. The grids differ in size, but each one contains values on both sides of the point where the cost reaches the budget, and
that point is all the rule needs. At the selected strengths, the cost is 1.54 nats per frame against a budget of 1.76 on
Moshi, 0.60 against 0.61 on CosyVoice3, and 2.45 against 2.46 on MOSS-TTS. A coarser grid can only select a smaller $\delta$ than a finer one, which works against our
watermark. The ablations of Table~\ref{tab:ablation} use the same rule on their own grids.

\paragraph{Predicted cost.}
By Proposition~\ref{prop:cost}, when the sampler only divides the logits by a temperature $T$, the first-order cost is
$\tfrac12(\delta/T)^2\sum_{s\in\mathcal S}a_s^\top B_sa_s$ nats per frame. The sum is at most $|\mathcal S|$ because of the
constraint in Eq.~\eqref{eq:solve}, and our solutions reach this limit on every stream except the first stream of Moshi,
where an amplitude bound is active ($a_1^\top B_1a_1=0.92$). On Moshi ($T=0.8$), the prediction is 1.50 nats per frame at
$\delta=0.7$, against 1.54 measured. On MOSS-TTS ($T=1.7$), it is 2.24 at $\delta=0.9$, against 2.45 measured. On
CosyVoice3 ($T=1$, one stream), it is 0.41 where 0.60 is measured, so the prediction would reach KGW's budget only at
$\delta\approx1.1$, whereas the measured cost reaches it at $\delta=0.9$ (Figure~\ref{fig:tts_cost}). On CosyVoice3, the
difference comes from the sampler, which keeps only the most likely tokens, redraws a token that repeats too often, and
renormalizes over speech tokens after the bias is added. The first-order cost is therefore a guide to the design, and
the strength must be chosen from the cost measured on validation generations, as we do for every model.

\begin{figure}[h]
\centering
\begin{subfigure}[t]{0.355\linewidth}
\centering
\includegraphics[width=\linewidth]{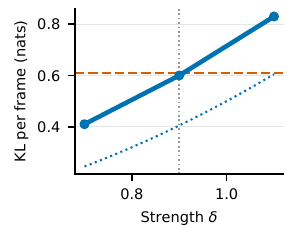}
\caption{CosyVoice3}
\label{fig:tts_cost_a}
\end{subfigure}\hspace{0.03\linewidth}
\begin{subfigure}[t]{0.355\linewidth}
\centering
\includegraphics[width=\linewidth]{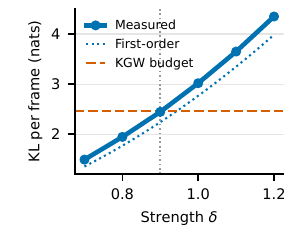}
\caption{MOSS-TTS}
\label{fig:tts_cost_b}
\end{subfigure}
\caption{\textbf{Cost of our watermark on the TTS models.} The plots show the measured cost against $\delta$ on the
validation set, together with the first-order prediction, KGW's budget and the selected $\delta$.}
\label{fig:tts_cost}
\end{figure}

\subsection{Attacks}
\label{app:attacks}
Every attack operates on the 24\,kHz waveform and keeps the length of the clip.

\paragraph{Codec attacks.} A codec pass encodes and decodes the
waveform with the codec of Table~\ref{tab:codec_list}, resampling to the codec's rate and back where it differs, and each
pass starts from the output of the previous one. The passes are stored as 32-bit floating point, so nothing is requantized
between them. SNAC's decoder is stochastic and is seeded per clip. SpeechTokenizer is a 16\,kHz model, but neither its
code nor the attack resamples the input, so it encodes the 24\,kHz waveform as it is, as in the attack suite
released with HiPT \citep{hipt}. This is the same for every method.

\paragraph{Signal-processing attacks.} The signal-processing attacks are applied once, with a seed per clip and attack,
and their settings are the rows of Table~\ref{tab:dsp}.
\begin{table}[h]
\caption{\textbf{Codecs.} The table lists the checkpoint and the setting of each codec used in the attacks.}
\label{tab:codec_list}
\centering\small
\begin{tabular}{l>{\raggedright\arraybackslash}p{0.42\linewidth}l}
\toprule
Codec & Implementation & Setting\\
\midrule
Mimi & Mimi of Moshi (\texttt{kyutai/moshiko-pytorch-bf16}) & 8 quantizers\\
CosyVoice3 native & speech tokenizer, flow-matching decoder and HiFT vocoder & clip as its own prompt\\
MOSS-TTS native & MOSS audio tokenizer & 32 quantizers\\
EnCodec & \texttt{facebook/encodec\_24khz} & 6\,kbps\\
SpeechTokenizer & \texttt{speechtokenizer\_hubert\_avg} & all quantizers, 24\,kHz input\\
SNAC & \texttt{hubertsiuzdak/snac\_24khz} & all quantizers\\
DAC16, DAC24, DAC44 & DAC 16, 24 and 44.1\,kHz models & all quantizers\\
Opus & ffmpeg, \texttt{libopus} & 24\,kbps\\
MP3 & ffmpeg, \texttt{libmp3lame} & 64\,kbps\\
AAC & ffmpeg, native AAC encoder & 64\,kbps\\
\bottomrule
\end{tabular}
\end{table}

\subsection{Baselines}
\label{app:baselines}
\paragraph{KGW.}
KGW adds $\delta=2$ to the logits of a green list that holds a quarter of the vocabulary. As in WMAR \citep{wmar}, one
green list, drawn once from a fixed seed, serves every frame, so the list does not depend on the previous tokens. On every
language model, KGW runs on the same streams as our watermark. Its detector counts the green tokens among the recovered
tokens, counting each distinct token once per stream, pools the counts over the watermarked streams, and computes a
$z$-score. The $z$-score is not centered on unwatermarked audio, so its threshold is set on the calibration answers like
ours.

\paragraph{WMAR.}
WMAR keeps KGW's watermark and fine-tunes Mimi so that more tokens survive decoding and re-encoding. We apply the fine-tuned
weights released by its authors, with (FT+Augs) and without (FT) augmentations during fine-tuning, to the stock Mimi
checkpoint. A WMAR answer is the KGW answer to the same prompt, with the same tokens, decoded by the fine-tuned decoder, so
WMAR has exactly KGW's cost. Its detector re-encodes the audio with the matching fine-tuned encoder and then applies KGW's
detector. Because the fine-tuned decoder changes the audio, the thresholds and the unwatermarked reference for quality come
from unwatermarked answers decoded and re-encoded by the same fine-tuned codec.

\paragraph{Post-hoc methods.}
AudioSeal, WavMark, Timbre and CRAW embed their watermark into the unwatermarked test answers with the released weights.
AudioSeal is scored by the fraction of frames it flags, and the other three by the bit accuracy of a fixed message. Their
thresholds follow the same rule as ours, on the same calibration clips. Bit-accuracy scores take few distinct values, so
the achieved calibration FPR lies between 0.08 and 0.9\,\% rather than at 1\,\%.

\paragraph{Aligned-IS.}
\label{app:alignedis}
Aligned-IS \citep{alignedis} is a distortion-free token-level watermark for speech language models. It groups the
vocabulary into $h$ clusters of nearly equal size by k-means on the token embeddings, so that retokenization keeps a token
within its cluster more often than it keeps the token itself. At each step, the key and the previous token select one
cluster, and the sampler moves probability toward that cluster in a way that leaves the sampling distribution unchanged on
average over keys. The detector counts the tokens that fall in their selected cluster.

The released code covers SpiritLM and SpeechGPT, which generate a single stream of about 500 units, so we ported it to
Moshi. The port follows the released code in the clustering, the keying, the reweighting and the detector, and agrees with
it in unit tests to numerical precision. As in the released code, the reweighting is applied before temperature and top-$k$.
For Moshi, we fit one clustering per stream on Mimi's codebook vectors, watermark all eight streams with a separate key per
stream, and pool the counts over the streams. The threshold is set on the calibration answers, as for every other method.

Table~\ref{tab:alignedis} gives the results. With the authors' setting of $h=20$ clusters, Aligned-IS is detected in
4.7\,\% of the test answers at identity, where 1.7\,\% of the unwatermarked answers are flagged, and in 2.0\,\% after eight
Mimi passes. As expected of a distortion-free method, it leaves quality unchanged, with a paired UTMOSv2 difference of
$+0.02$ on the test set. On the validation set, more clusters raise the detection at identity, to 32.0\,\% with 50 clusters and
50.0\,\% with 100, but after eight Mimi passes every setting stays at or below 2.4\,\%. Four variants of the setting with $h=20$ do not change this either. They
read only streams 1--4, apply the reweighting after temperature and top-$k$, cluster the input embeddings of the language
model instead of Mimi's codebook vectors, or key the selection on the clusters of the previous two tokens. The second and the last
variants are our own modifications. On SpiritLM, the original paper reports a TPR of 96\,\% without attack, close to the 99\,\% of
KGW. Our setting differs from theirs in several ways. Moshi generates eight streams of 2,048 tokens rather than one of about
500, it samples at temperature 0.8 with top-$k=250$, and the previous token, which keys the selection, survives Mimi's
retokenization in only 17--72\,\% of the counted frames, depending on the stream (Table~\ref{tab:selfloops}). We have not
established which of these differences explains the gap. Because Aligned-IS stays close to its false-positive rate on
Moshi, we do not include it in the main comparison.
\begin{table}[h]
\caption{\textbf{Aligned-IS on Moshi.} TPR (\%) at the fixed threshold without attack (Identity), after one and eight
Mimi passes, and after eight EnCodec passes. FPR is the rate on the unwatermarked answers of the same set without attack.
The indented variants use $h=20$, and KGW is shown for reference.}
\label{tab:alignedis}
\centering\small\setlength{\tabcolsep}{4pt}
\begin{tabular}{llccccc}
\toprule
Set & Method & FPR & Identity & Mimi $\times1$ & Mimi $\times8$ & EnCodec $\times8$\\
\midrule
Test & Aligned-IS, $h=20$ & 1.7 & 4.7 & 3.3 & 2.0 & 0.5\\
Test & KGW & 1.5 & 78.5 & 64.5 & 8.3 & 16.5\\
\midrule
Validation & Aligned-IS, $h=20$ & 2.0 & 4.4 & 4.0 & 2.4 & 0.4\\
Validation & Aligned-IS, $h=50$ & 1.2 & 32.0 & 9.2 & 2.4 & 1.6\\
Validation & Aligned-IS, $h=100$ & 0.8 & 50.0 & 14.0 & 1.2 & 6.4\\
Validation & \quad streams 1--4 only & 1.2 & 3.2 & 2.0 & 2.0 & 0.4\\
Validation & \quad reweighting after top-$k$ & 2.0 & 2.4 & 0.8 & 1.2 & 0.4\\
Validation & \quad language-model embeddings & 2.0 & 31.2 & 8.0 & 0.0 & 0.0\\
Validation & \quad keyed on two previous clusters & 0.4 & 10.8 & 2.8 & 0.8 & 0.8\\
Validation & KGW & 0.8 & 80.8 & 63.2 & 8.0 & 15.6\\
\bottomrule
\end{tabular}
\end{table}

\subsection{Speech quality}
\label{app:qualityprotocol}
\paragraph{Quality scoring.}
For each clip, the speech level is the 95th percentile of 20\,ms frame levels. Leading and trailing regions more than
35\,dB below it, or below $-60$\,dBFS, are removed with a 200\,ms margin. Silence inside the clip is kept. The trim is
applied only before quality scoring. UTMOSv2 averages five seeded evaluations per clip.

\paragraph{Empty answers.}
An answer of Moshi is empty when its text stream contains no word once byte tokens such as \texttt{<0x14>} are removed.
All such answers end at the minimum length of 38 frames (3\,s). Empty answers make up 15.3\,\% of the unwatermarked test
answers, 12.8\,\% for KGW and WMAR, which share KGW's tokens, and 7.5\,\% for ours, and 13.1\,\% of the calibration clips.

\FloatBarrier

\section{Additional results}
\label{app:results}
This appendix reports the codecs that preserve the waveform (Appendix~\ref{app:codecs}), the TPR after every pass
(Appendix~\ref{app:passes}), the signal-processing attacks (Appendix~\ref{app:dsp}), the false positives
(Appendix~\ref{app:fpr}), the effect of clip length (Appendix~\ref{app:duration}) and of empty answers
(Appendix~\ref{app:nonempty}), the paired differences in speech quality (Appendix~\ref{app:quality}), and the full
ablation (Appendix~\ref{app:ablation}).

\subsection{High-fidelity codecs}
\label{app:codecs}
Table~\ref{tab:hifi} lists the codecs that preserve the waveform closely. The post-hoc methods, which embed into the
waveform, are strongest here, although WavMark is lost after the DAC models and Opus, and Timbre after DAC44. Our
watermark keeps at least 90.7\,\% on every model and codec here except DAC24 on Moshi (64.0\,\%), and KGW and WMAR stay
below it in every cell.

\begin{table}[h]
\caption{\textbf{Detection after eight passes through high-fidelity codecs.} TPR (\%) on the test sets at the fixed
threshold, $\pm$ the half-width of the Wilson 95\,\% interval. DAC24 and DAC44 are the 24 and 44.1\,kHz DAC
models. Opus runs at 24\,kbps, and MP3 and AAC at 64\,kbps. Bold is as defined in Section~\ref{sec:experiments}, within
each model. None of these attacks flags more than 5\,\% of unwatermarked answers, so no cell is excluded from the bold as
in Table~\ref{tab:moshi_codecs}.}
\label{tab:hifi}
\centering\small\setlength{\tabcolsep}{4pt}
\input{tables/r_a_hifi_codecs}
\end{table}

\subsection{Detection after every pass}
\label{app:passes}
Figures~\ref{fig:passes_moshi}, \ref{fig:passes_cosyvoice3} and~\ref{fig:passes_moss} show the TPR after every pass of every
codec. Our watermark loses detection gradually from pass to pass rather than at the first pass, under every codec
and on every model. AudioSeal's TPR under repeated DAC16 and DAC24 passes rises and falls from one pass to the next.
This pattern has the same shape on all three models, and AudioSeal's FPR on unwatermarked answers stays flat over the same
passes, so it reflects how its watermark score responds to repeated DAC resynthesis rather than noise.

\begin{figure}[p]
\centering
\includegraphics[width=\linewidth]{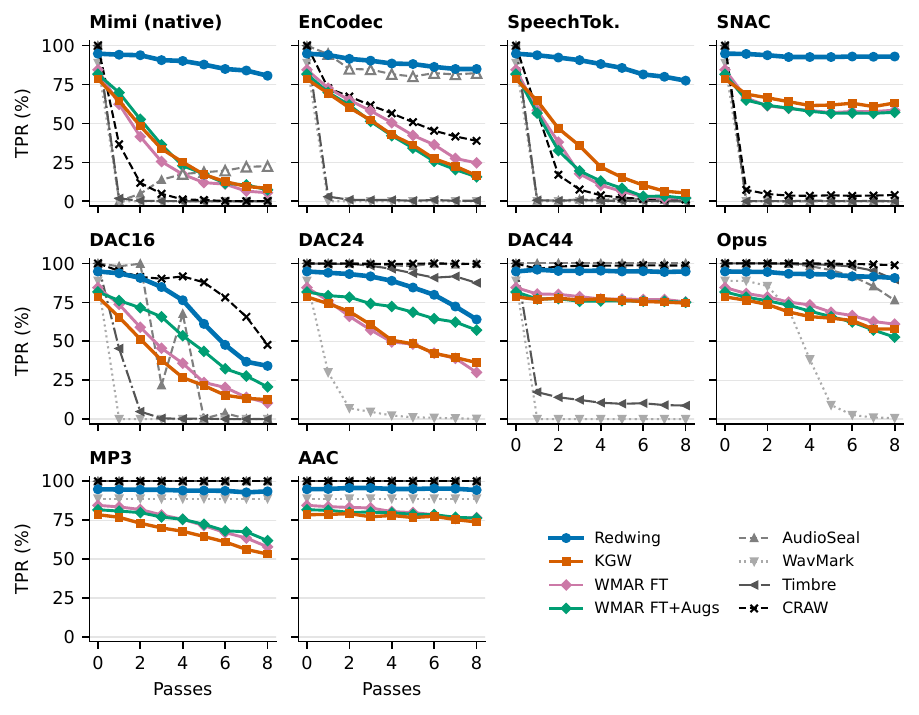}
\caption{\textbf{Detection on Moshi after every pass of each codec.} TPR (\%) at the fixed threshold on the test
set, where pass 0 is the identity condition. Open markers show where the same attack also flags more than 5\,\% of
unwatermarked answers (Table~\ref{tab:fpr}).}
\label{fig:passes_moshi}
\end{figure}

\begin{figure}[p]
\centering
\includegraphics[width=\linewidth]{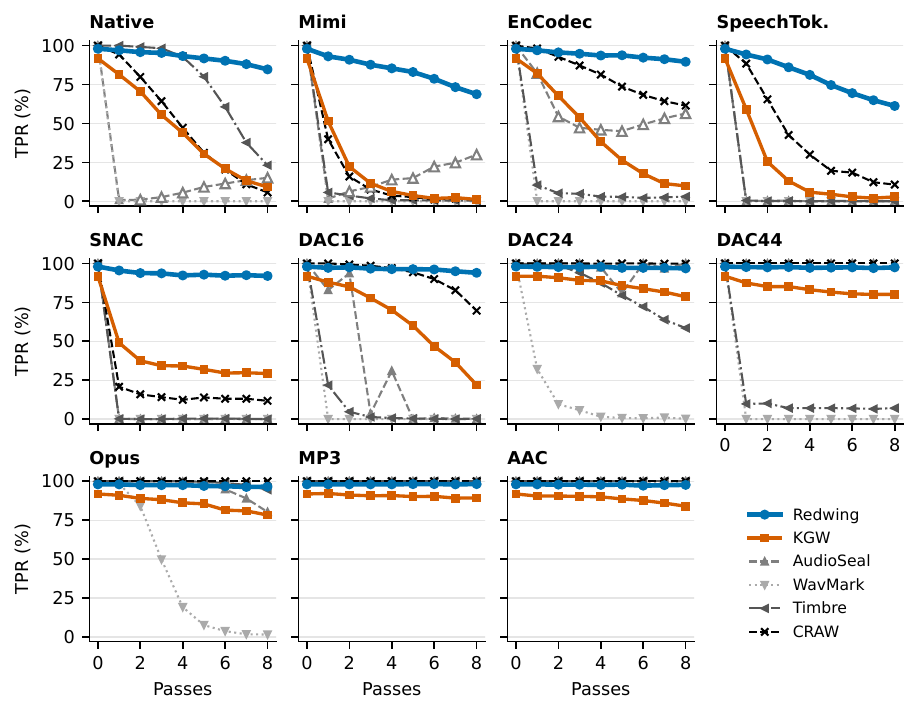}
\caption{\textbf{Detection on CosyVoice3 after every pass of each codec.} TPR (\%) at the fixed threshold on the test
set, where pass 0 is the identity condition. Native is the model's own codec. Open markers show where the same attack also flags more than 5\,\% of
unwatermarked answers (Table~\ref{tab:fpr}).}
\label{fig:passes_cosyvoice3}
\end{figure}

\begin{figure}[p]
\centering
\includegraphics[width=\linewidth]{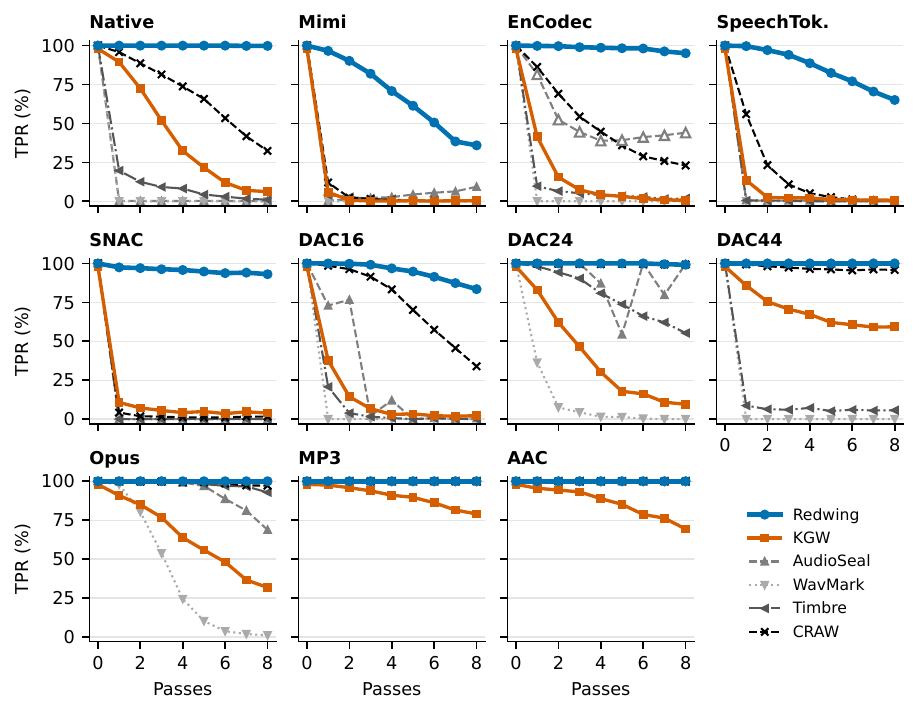}
\caption{\textbf{Detection on MOSS-TTS after every pass of each codec.} TPR (\%) at the fixed threshold on the test
set, where pass 0 is the identity condition. Native is the model's own codec. Open markers show where the same attack also flags more than 5\,\% of
unwatermarked answers (Table~\ref{tab:fpr}).}
\label{fig:passes_moss}
\end{figure}

\subsection{Signal-processing attacks}
\label{app:dsp}
Table~\ref{tab:dsp} lists the TPR after each signal-processing attack at the fixed threshold. The attacks fall
into two groups. The first group changes the sound but keeps its timing, as noise, filtering, quantization, gain and
reverberation do. Under these attacks most watermarks remain detectable, and the post-hoc methods, which are trained
against such distortions, are usually the strongest. \ours stays close to its TPR without attack under the milder
settings, and it loses the most under 10\,dB noise and 4-bit quantization on Moshi. Pitch shifts are the
exception in this group. They remove every post-hoc watermark, while ours keeps between 67 and 100\,\%.

The second group moves the audio in time. Our detector searches over offsets of up to two frames, so a shift by a whole
frame costs nothing, as for 80\,ms on Moshi and MOSS-TTS and 40\,ms on CosyVoice3, whose frames are 40\,ms long. A shift
by a fraction of a frame weakens the watermark strongly on Moshi and MOSS-TTS, and less on CosyVoice3. Cropping removes
an unknown number of frames from the start of the clip, which the offset search does not cover, and it removes almost
all of our watermark. Speed changes of 2\,\% or more remove most of every watermark except KGW's on CosyVoice3. KGW and
WMAR suffer from the same time-domain attacks as ours, while the post-hoc methods are unaffected by shifts and crops.
Under all of these attacks, the FPR of the token-level detectors on unwatermarked answers stays at or below 3.8\,\%,
that of Timbre and CRAW at or below 2.2\,\%, and that of WavMark at or below 0.2\,\%. AudioSeal's threshold, however, also flags many unwatermarked answers
after quantization, between 6.8 and 99.3\,\% depending on the model and the bit depth, and 6.5\,\% after smoothing on
CosyVoice3. Its TPR under these attacks therefore says little about its watermark, and we mark these cells in
Table~\ref{tab:dsp}.

\begin{table}[p]
\caption{\textbf{Detection after signal-processing attacks.} TPR (\%) on the test sets at the fixed threshold,
rounded to whole percent. The Wilson 95\,\% interval of every value is at most 8 percentage points wide. Within each row and model,
bold marks the highest TPR and every TPR whose interval overlaps it, and nothing when all overlap. $^\ast$The same attack
also flags more than 5\,\% of unwatermarked answers, so the cell is excluded from the bold. WavMark does not embed into 67 of Moshi's answers, which caps its TPR there at 89\,\%.}
\label{tab:dsp}
\centering\scriptsize\setlength{\tabcolsep}{2.5pt}
\input{tables/r_a_dsp}
\end{table}

\subsection{False positives}
\label{app:fpr}
\paragraph{Unwatermarked test answers.} Table~\ref{tab:fpr} gives the FPR on the unwatermarked test answers after each attack, at each method's fixed
threshold. For every token-domain method and for WavMark, Timbre and CRAW, it stays at or below 2.5\,\% under every
codec, within the sampling noise of 600 clips (Appendix~\ref{app:data}). AudioSeal is the exception. Some codecs
raise its score on unwatermarked audio as well, so that after eight passes its threshold flags 10.5\,\% of Moshi's
answers under Mimi, 15.1 and 20.8\,\% of CosyVoice3's under Mimi and its own codec, and between 45.9 and 75.3\,\% of the
answers of all three models under EnCodec. Under these codecs, the TPR of AudioSeal in Tables~\ref{tab:moshi_codecs}
and~\ref{tab:tts_codecs} therefore says little about its watermark.

\paragraph{Human recordings.} Our Moshi detector also rarely fires on human speech. On the 1,000 WildVoice recordings that serve as Moshi's prompts,
it flags 0.5\,\% without attack and at most 1.5\,\% after any of the 116 attack conditions, at the same fixed threshold.
KGW flags 0.7\,\% and at most 1.2\,\%. We scored the post-hoc detectors on the same recordings without attack and after
eight passes of each of the ten codecs (Table~\ref{tab:human}). On these eleven conditions, ours and KGW flag at most 0.8
and 1.0\,\%, and WavMark, Timbre and CRAW at most 1.6\,\%. AudioSeal's threshold, in contrast, does not carry over to human
speech. It flags 4.3\,\% of the recordings without attack, against 0.7\,\% of Moshi's unwatermarked answers, and 42.2, 79.1
and 7.0\,\% after eight passes of Mimi, EnCodec and Opus.
\begin{table}[h]
\caption{\textbf{False positives on human speech.} FPR (\%) on the 1,000 WildVoice recordings that serve as Moshi's prompts,
at each method's fixed threshold for Moshi, without attack (Identity) and after eight passes of each codec. $^\ast$Values
above 5\,\%.}
\label{tab:human}
\centering\small\setlength{\tabcolsep}{3pt}
\resizebox{\linewidth}{!}{\input{tables/r_a_human}}
\end{table}

\begin{table}[h]
\caption{\textbf{False positives after attack.} FPR (\%) on the unwatermarked test answers at each method's fixed
threshold, without attack (Identity) and after eight passes. Native is each model's own codec, which is Mimi on
Moshi. $^\ast$Values above 5\,\%.}
\label{tab:fpr}
\centering\small\setlength{\tabcolsep}{3pt}
\resizebox{\linewidth}{!}{\input{tables/r_a_fpr}}
\end{table}

\subsection{Clip length}
\label{app:duration}
The detection score sums over the frames of a clip, so a short clip carries less evidence. Table~\ref{tab:duration} splits
Moshi's test answers into bins of 50 frames, or 4\,s, of length. The first bin starts at 38 frames (3\,s), the minimum
length of an answer, and the last ends at the limit of 200 frames (16\,s). Without attack, our watermark is detected in
every answer longer than 100 frames (8\,s), and 29 of its 31 misses fall in the shortest bin. After eight Mimi passes, our
TPR rises from 39.5\,\% in the shortest bin to 100\,\% in the longest, whereas KGW stays at or below 13.5\,\% at every
length. The FPR on unwatermarked answers stays at or below 2.2\,\% in every bin.

\begin{table}[h]
\caption{\textbf{Detection on Moshi by answer length.} TPR (\%) at the fixed threshold on the test set $\pm$ the
95\,\% Wilson half-width, by answer length in frames. Each bin covers 50 frames, or 4\,s, and All is the TPR on all
test answers.}
\label{tab:duration}
\centering\small\setlength{\tabcolsep}{8pt}
\input{tables/r_a_duration}
\end{table}

\subsection{Non-empty answers}
\label{app:nonempty}
\paragraph{Detection.} Removing Moshi's empty answers from the watermarked test answers, at the same thresholds, raises the identity TPR of every
method: ours from 94.8 to 97.1\,\%, KGW from 78.5 to 88.3\,\%, and WMAR from 84.5 to 94.8\,\% (FT) and from 81.7 to
91.4\,\% (FT+Augs). After eight Mimi passes, ours rises from 80.7 to 83.8\,\% and KGW from 8.3 to 12.2\,\%. The FPR on
non-empty unwatermarked answers stays at 1.2\,\% for ours and 1.4\,\% for KGW. Re-deriving the thresholds from the
non-empty calibration clips moves the TPR by at most a few points, for example KGW's identity TPR from 88.3 to 91.0\,\%.
Empty answers therefore lower the TPR of KGW and WMAR more than ours, but do not change the comparison after resynthesis.

\paragraph{Quality.} On all answers, the quality predictors give our watermark a higher mean than the unwatermarked model (UTMOSv2 3.00 against
2.91), because it produces half as many empty answers and its empty answers contain speech-like audio (UTMOSv2 2.36 against
1.60). On prompts where the answers are non-empty, our watermark and KGW lower UTMOSv2 by about the same amount
(Appendix~\ref{app:quality}).

\subsection{Speech quality}
\label{app:quality}
Table~\ref{tab:quality_posthoc} gives the quality of the post-hoc methods, which do not change the sampling distribution
and therefore have no cost.

\begin{table}[h]
\caption{\textbf{Speech quality of the post-hoc methods.} Mean $\pm$ 95\,\% CI on the test sets, as in
Table~\ref{tab:quality}. Bold as in Table~\ref{tab:quality}.}
\label{tab:quality_posthoc}
\centering\small\setlength{\tabcolsep}{4pt}
\input{tables/r_a_quality_posthoc}
\end{table}

Table~\ref{tab:quality_delta} gives the difference of each method from the unwatermarked answer to the same prompt. On
Moshi, every row uses the 448 test prompts on which the unwatermarked model, ours and KGW all gave a non-empty answer. WMAR
shares KGW's tokens and the post-hoc methods watermark the unwatermarked answers, so these prompts are non-empty for them as
well. WMAR is compared with the unwatermarked answers decoded by its own fine-tuned decoder, so that its row shows the
effect of the watermark alone. On Moshi, our watermark lowers UTMOSv2 by 0.09, about as much as KGW and WMAR, and raises
DNSMOS Pro by 0.10. Section~\ref{sec:results} reports a drop of 0.10, which compares the means over each method's own
non-empty answers (Table~\ref{tab:quality}). Here, both answers of a pair come from the same prompt, which gives 0.09. On CosyVoice3, it changes every predictor by at most 0.01. On MOSS-TTS, it lowers UTMOSv2 by 0.12,
against 0.09 for KGW, and every other predictor by at most 0.04. Among the post-hoc methods, CRAW lowers every predictor on every model except NISQA-TTS on Moshi, and WavMark
lowers UTMOSv2 on Moshi by 0.27.

\begin{table}[h]
\caption{\textbf{Paired differences in speech quality.} Mean difference from the unwatermarked answer to the same prompt
$\pm$ the 95\,\% CI on the test sets. On Moshi, only prompts with non-empty answers from the unwatermarked model, ours and
KGW are used, and on the TTS models all clips. $^\dagger$Compared with the unwatermarked answers decoded by the same
fine-tuned decoder.}
\label{tab:quality_delta}
\centering\footnotesize\setlength{\tabcolsep}{3pt}
\input{tables/r_a_quality_delta}
\end{table}

\subsection{Ablations}
\label{app:ablation}
Table~\ref{tab:ablation} gives the full ablation of Section~\ref{sec:results}, with the TPR after every foreign codec and
the quality and cost of each variant.

\begin{table}[h]
\caption{\textbf{Ablations on Moshi.} TPR (\%) at the fixed threshold on the test set $\pm$ the 95\,\% Wilson half-width, and
UTMOSv2 on non-empty answers $\pm$ the 95\,\% CI. Each variant removes one component and runs at its own budget-matched
$\delta$, and KL is its cost per frame. The random basis uses 16 random orthonormal directions, and the transition basis
the leading singular vectors of the transition matrix $P$. The random function is one fixed random $g=h$ per stream,
without basis or solve. $^\ddagger$The same watermarked answers as
ours, but the detector scores them with the embedding function $g$ in place of the solved detection function $h$, which
tests whether a separate detection function is needed. Bold is as defined in Section~\ref{sec:experiments}.}
\label{tab:ablation}
\centering\footnotesize\setlength{\tabcolsep}{2.5pt}
\input{tables/r_t3_ablation}
\end{table}

\FloatBarrier

\section{Extended related work}
\label{app:related}
\paragraph{Post-hoc audio watermarks.}
Most audio watermarks are embedded after generation. WavMark \citep{wavmark}, AudioSeal \citep{audioseal}, Timbre
\citep{timbre} and SilentCipher \citep{singh2024silentcipher} train an embedding network and a detection network jointly, with a set of distortions applied between the
two, so that the watermark survives common edits. CRAW \citep{craw} adds neural-codec resynthesis to these distortions. Latent-Mark \citep{latentmark} targets neural codecs
differently. It optimizes the waveform so that the codec's latent representation shifts in a detectable direction, across
several surrogate codecs.
Post-hoc methods apply to the output of any generator, and they keep their
watermark under codecs that preserve the waveform closely. Benchmarks find that their robustness varies widely across
edits and compression \citep{liu2024audiomarkbench,wen2025sok}, and in our experiments repeated resynthesis through
low-bitrate neural codecs removes most of their watermark (Section~\ref{sec:results}).

\paragraph{Generation-time watermarks.}
Generation-time watermarks are embedded while the model generates, instead of into the finished output. Latent
watermarking \citep{sanroman2024latent} trains the generator on watermarked audio, so that its outputs carry the
watermark. Token-level watermarks instead bias the logits of the tokens that a language model samples.
KGW \citep{kgw} watermarks a text language model by adding a bias to the logits of a keyed subset of the vocabulary, the
green list, and detects an excess of green tokens with a $z$-test. Later work makes the watermark distortion-free or
undetectable \citep{kuditipudi2024robust,hu2024unbiased,christ2024undetectable}, studies its robustness to edits of the
text \citep{zhao2024provable,kirchenbauer2024reliability}, and deploys it in production \citep{dathathri2024synthid}. \citet{liu2024survey} survey this line of work. Because speech language models generate codec tokens,
the same scheme applies to them, but its detector sees only the tokens recovered from the audio. Aligned-IS
\citep{alignedis} is a distortion-free token-level watermark for audio generation. WMAR \citep{wmar} keeps KGW and instead
fine-tunes the codec so that more tokens survive decoding and re-encoding. HiPT \citep{hipt} is closest to our work. It
clusters the vocabulary by how often retokenization confuses two tokens and applies KGW to the cluster labels, so that a
substitution within a cluster leaves the watermark intact. Our basis also comes from the substitutions of retokenization,
but instead of a hard partition it gives every token a real value that varies smoothly over the substitution graph.
In KGW and HiPT, the green list plays two roles. The bias raises the logits of the green tokens, and the detector counts
them. Our embedding and detection functions play the same two roles, but they are real-valued, differ from each other,
and are solved for rather than drawn at random.

\paragraph{Neural codecs and speech language models.}
Neural codecs such as SoundStream \citep{zeghidour2022soundstream}, EnCodec \citep{encodec}, DAC \citep{dac}, SpeechTokenizer \citep{speechtokenizer}, SNAC
\citep{snac} and Mimi \citep{moshi} compress speech into streams of discrete tokens with residual vector quantization, a multi-stage form of the vector
quantization of \citet{vandenoord2017vqvae}, or with variants of it. Mimi, for example, quantizes its first level separately
and distills semantic information into it. Audio language models generate such tokens directly
\citep{borsos2023audiolm,wang2023valle,copet2023musicgen}. Codecs are trained to reconstruct the waveform rather than to
reproduce their own tokens, and re-encoding a decoded clip changes some of them.

\paragraph{Spectral bases on graphs.}
We turn these substitutions into a graph
over the vocabulary. The eigenvectors of its normalized Laplacian with the smallest eigenvalues vary slowly across strongly
connected tokens \citep{chung1997spectral}. This is the property on which spectral clustering and Laplacian eigenmaps
rely \citep{shi2000normalized,belkin2003laplacian,luxburg2007tutorial} and the usual notion of a smooth signal on a graph
\citep{shuman2013emerging}, and these eigenvectors form our basis.

%% file: tables/r_a_hifi_codecs.tex
\begin{tabular}{lrrrrr}
\toprule
Method & DAC24 & DAC44 & Opus & MP3 & AAC \\
\midrule
\multicolumn{6}{l}{\textbf{Moshi}} \\
Redwing & 64.0\,{\scriptsize$\pm$3.8} & 94.8\,{\scriptsize$\pm$1.8} & 90.7\,{\scriptsize$\pm$2.3} & 93.5\,{\scriptsize$\pm$2.0} & 94.3\,{\scriptsize$\pm$1.9} \\
KGW & 36.3\,{\scriptsize$\pm$3.8} & 74.5\,{\scriptsize$\pm$3.5} & 57.8\,{\scriptsize$\pm$3.9} & 53.2\,{\scriptsize$\pm$4.0} & 73.8\,{\scriptsize$\pm$3.5} \\
WMAR FT & 30.0\,{\scriptsize$\pm$3.7} & 75.3\,{\scriptsize$\pm$3.4} & 61.0\,{\scriptsize$\pm$3.9} & 57.8\,{\scriptsize$\pm$3.9} & 76.7\,{\scriptsize$\pm$3.4} \\
WMAR FT+Augs & 57.2\,{\scriptsize$\pm$3.9} & 75.0\,{\scriptsize$\pm$3.5} & 52.7\,{\scriptsize$\pm$4.0} & 61.8\,{\scriptsize$\pm$3.9} & 76.3\,{\scriptsize$\pm$3.4} \\
\midrule
AudioSeal & \textbf{100.0}\,{\scriptsize$\pm$0.3} & \textbf{100.0}\,{\scriptsize$\pm$0.3} & 76.3\,{\scriptsize$\pm$3.4} & \textbf{100.0}\,{\scriptsize$\pm$0.3} & \textbf{100.0}\,{\scriptsize$\pm$0.3} \\
WavMark & 0.2\,{\scriptsize$\pm$0.5} & 0.0\,{\scriptsize$\pm$0.3} & 0.7\,{\scriptsize$\pm$0.7} & 88.8\,{\scriptsize$\pm$2.5} & 88.8\,{\scriptsize$\pm$2.5} \\
Timbre & 87.5\,{\scriptsize$\pm$2.6} & 8.7\,{\scriptsize$\pm$2.3} & 89.7\,{\scriptsize$\pm$2.4} & \textbf{100.0}\,{\scriptsize$\pm$0.3} & \textbf{100.0}\,{\scriptsize$\pm$0.3} \\
CRAW & \textbf{99.7}\,{\scriptsize$\pm$0.6} & 98.7\,{\scriptsize$\pm$1.0} & \textbf{98.8}\,{\scriptsize$\pm$0.9} & \textbf{100.0}\,{\scriptsize$\pm$0.3} & \textbf{100.0}\,{\scriptsize$\pm$0.3} \\
\midrule[\heavyrulewidth]\addlinespace[2pt]
\multicolumn{6}{l}{\textbf{CosyVoice3}} \\
Redwing & 96.8\,{\scriptsize$\pm$1.4} & 97.5\,{\scriptsize$\pm$1.3} & 96.5\,{\scriptsize$\pm$1.5} & 98.2\,{\scriptsize$\pm$1.1} & 97.7\,{\scriptsize$\pm$1.2} \\
KGW & 78.6\,{\scriptsize$\pm$3.3} & 80.1\,{\scriptsize$\pm$3.2} & 78.3\,{\scriptsize$\pm$3.3} & 89.3\,{\scriptsize$\pm$2.5} & 83.9\,{\scriptsize$\pm$2.9} \\
\midrule
AudioSeal & \textbf{100.0}\,{\scriptsize$\pm$0.3} & \textbf{100.0}\,{\scriptsize$\pm$0.3} & 80.2\,{\scriptsize$\pm$3.2} & \textbf{100.0}\,{\scriptsize$\pm$0.3} & \textbf{100.0}\,{\scriptsize$\pm$0.3} \\
WavMark & 0.2\,{\scriptsize$\pm$0.5} & 0.0\,{\scriptsize$\pm$0.3} & 1.5\,{\scriptsize$\pm$1.0} & \textbf{99.7}\,{\scriptsize$\pm$0.6} & \textbf{99.7}\,{\scriptsize$\pm$0.6} \\
Timbre & 58.5\,{\scriptsize$\pm$3.9} & 7.0\,{\scriptsize$\pm$2.1} & 94.5\,{\scriptsize$\pm$1.8} & \textbf{100.0}\,{\scriptsize$\pm$0.3} & \textbf{100.0}\,{\scriptsize$\pm$0.3} \\
CRAW & \textbf{99.8}\,{\scriptsize$\pm$0.5} & \textbf{100.0}\,{\scriptsize$\pm$0.3} & \textbf{100.0}\,{\scriptsize$\pm$0.3} & \textbf{100.0}\,{\scriptsize$\pm$0.3} & \textbf{100.0}\,{\scriptsize$\pm$0.3} \\
\midrule[\heavyrulewidth]\addlinespace[2pt]
\multicolumn{6}{l}{\textbf{MOSS-TTS}} \\
Redwing & \textbf{99.2}\,{\scriptsize$\pm$0.8} & \textbf{100.0}\,{\scriptsize$\pm$0.3} & \textbf{100.0}\,{\scriptsize$\pm$0.3} & \textbf{100.0}\,{\scriptsize$\pm$0.3} & \textbf{100.0}\,{\scriptsize$\pm$0.3} \\
KGW & 9.5\,{\scriptsize$\pm$2.4} & 59.3\,{\scriptsize$\pm$3.9} & 31.7\,{\scriptsize$\pm$3.7} & 79.0\,{\scriptsize$\pm$3.3} & 69.3\,{\scriptsize$\pm$3.7} \\
\midrule
AudioSeal & \textbf{99.7}\,{\scriptsize$\pm$0.6} & \textbf{100.0}\,{\scriptsize$\pm$0.3} & 68.8\,{\scriptsize$\pm$3.7} & \textbf{100.0}\,{\scriptsize$\pm$0.3} & \textbf{100.0}\,{\scriptsize$\pm$0.3} \\
WavMark & 0.0\,{\scriptsize$\pm$0.3} & 0.0\,{\scriptsize$\pm$0.3} & 1.0\,{\scriptsize$\pm$0.9} & \textbf{99.7}\,{\scriptsize$\pm$0.6} & \textbf{99.7}\,{\scriptsize$\pm$0.6} \\
Timbre & 55.3\,{\scriptsize$\pm$4.0} & 5.7\,{\scriptsize$\pm$1.9} & 92.8\,{\scriptsize$\pm$2.1} & \textbf{100.0}\,{\scriptsize$\pm$0.3} & \textbf{100.0}\,{\scriptsize$\pm$0.3} \\
CRAW & \textbf{99.3}\,{\scriptsize$\pm$0.7} & 95.8\,{\scriptsize$\pm$1.6} & 97.2\,{\scriptsize$\pm$1.4} & \textbf{100.0}\,{\scriptsize$\pm$0.3} & \textbf{100.0}\,{\scriptsize$\pm$0.3} \\
\bottomrule
\end{tabular}

%% file: tables/r_a_dsp.tex
\begin{tabular}{lrrrrrrrr|rrrrrr|rrrrrr}
\toprule
 & \multicolumn{8}{c}{Moshi} & \multicolumn{6}{c}{CosyVoice3} & \multicolumn{6}{c}{MOSS-TTS} \\ \cmidrule(lr){2-9}\cmidrule(lr){10-15}\cmidrule(lr){16-21}
Attack & \rotatebox{90}{Redwing} & \rotatebox{90}{KGW} & \rotatebox{90}{WMAR FT} & \rotatebox{90}{WMAR FT+Augs} & \rotatebox{90}{AudioSeal} & \rotatebox{90}{WavMark} & \rotatebox{90}{Timbre} & \rotatebox{90}{CRAW} & \rotatebox{90}{Redwing} & \rotatebox{90}{KGW} & \rotatebox{90}{AudioSeal} & \rotatebox{90}{WavMark} & \rotatebox{90}{Timbre} & \rotatebox{90}{CRAW} & \rotatebox{90}{Redwing} & \rotatebox{90}{KGW} & \rotatebox{90}{AudioSeal} & \rotatebox{90}{WavMark} & \rotatebox{90}{Timbre} & \rotatebox{90}{CRAW} \\
\midrule
None & 95 & 78 & 84 & 82 & \textbf{100} & 89 & \textbf{100} & \textbf{100} & 98 & 92 & \textbf{100} & \textbf{100} & \textbf{100} & \textbf{100} & \textbf{100} & 98 & \textbf{100} & \textbf{100} & \textbf{100} & \textbf{100} \\
\multicolumn{21}{l}{\textit{White noise (SNR)}} \\
\quad 40 dB & 95 & 79 & 76 & 82 & \textbf{100} & 88 & \textbf{100} & \textbf{100} & 98 & 91 & \textbf{100} & 97 & \textbf{99} & \textbf{100} & \textbf{100} & 83 & \textbf{100} & 95 & \textbf{99} & \textbf{100} \\
\quad 30 dB & 95 & 78 & 45 & 81 & \textbf{100} & 75 & 89 & \textbf{100} & 98 & 89 & \textbf{100} & 61 & 86 & \textbf{100} & \textbf{100} & 53 & \textbf{100} & 66 & 85 & \textbf{100} \\
\quad 20 dB & 87 & 47 & 1 & 80 & \textbf{100} & 3 & 31 & \textbf{99} & 97 & 84 & 98 & 4 & 39 & \textbf{100} & \textbf{100} & 14 & 97 & 7 & 40 & \textbf{99} \\
\quad 10 dB & 38 & 11 & 0 & 69 & 92 & 0 & 10 & \textbf{96} & 96 & 65 & 70 & 0 & 5 & \textbf{100} & \textbf{98} & 2 & 60 & 0 & 6 & \textbf{98} \\
\multicolumn{21}{l}{\textit{Filtering}} \\
\quad High-pass 500 Hz & 76 & 13 & 36 & 67 & \textbf{100} & 89 & \textbf{100} & \textbf{100} & 97 & 74 & \textbf{100} & \textbf{100} & \textbf{100} & \textbf{100} & \textbf{100} & 20 & \textbf{100} & \textbf{100} & \textbf{100} & \textbf{100} \\
\quad Low-pass 3.5 kHz & 95 & 74 & 82 & 80 & \textbf{100} & 89 & \textbf{100} & \textbf{100} & 98 & 88 & \textbf{100} & \textbf{100} & \textbf{100} & \textbf{100} & \textbf{100} & 66 & \textbf{100} & \textbf{100} & \textbf{100} & \textbf{100} \\
\quad Smoothing 2 ms & 73 & 4 & 1 & 30 & 43 & 22 & \textbf{93} & \textbf{96} & \textbf{98} & 80 & 63$^\ast$ & 14 & 93 & \textbf{99} & \textbf{99} & 1 & 55 & 23 & 88 & 88 \\
\multicolumn{21}{l}{\textit{Quantization}} \\
\quad 8 bit & \textbf{91} & 71 & 37 & 70 & 93$^\ast$ & 44 & 78 & \textbf{89} & 98 & 88 & 100$^\ast$ & 65 & 94 & \textbf{100} & \textbf{100} & 54 & 100$^\ast$ & 67 & 92 & \textbf{100} \\
\quad 6 bit & \textbf{78} & 23 & 2 & 15 & 87$^\ast$ & 0 & 19 & \textbf{80} & 96 & 75 & 100$^\ast$ & 1 & 57 & \textbf{99} & \textbf{99} & 13 & 100$^\ast$ & 1 & 50 & \textbf{97} \\
\quad 4 bit & \textbf{41} & 1 & 2 & 1 & 82$^\ast$ & 0 & 1 & 0 & \textbf{87} & 14 & 99$^\ast$ & 0 & 2 & 2 & \textbf{89} & 1 & 99$^\ast$ & 0 & 2 & 1 \\
\quad $\mu$-law 256 & 95 & 79 & 79 & 81 & \textbf{100} & 89 & \textbf{100} & \textbf{100} & 98 & 91 & \textbf{100} & \textbf{99} & \textbf{100} & \textbf{100} & \textbf{100} & 91 & \textbf{100} & 99 & \textbf{100} & \textbf{100} \\
\quad $\mu$-law 64 & 94 & 68 & 48 & 77 & 96 & 78 & 95 & \textbf{100} & 98 & 89 & \textbf{100} & 87 & \textbf{99} & \textbf{100} & \textbf{100} & 58 & \textbf{100} & 83 & \textbf{100} & \textbf{100} \\
\multicolumn{21}{l}{\textit{Gain}} \\
\quad +6 dB & 93 & 72 & 58 & 53 & \textbf{100} & 89 & \textbf{100} & \textbf{100} & 97 & 91 & \textbf{100} & \textbf{100} & \textbf{100} & \textbf{100} & \textbf{100} & 88 & \textbf{100} & \textbf{100} & \textbf{100} & \textbf{100} \\
\quad $-$6 dB & 92 & 48 & 46 & 43 & \textbf{100} & 89 & \textbf{100} & \textbf{100} & 98 & 92 & \textbf{100} & \textbf{100} & \textbf{100} & \textbf{100} & \textbf{100} & 85 & \textbf{100} & \textbf{100} & \textbf{100} & \textbf{100} \\
\multicolumn{21}{l}{\textit{Reverberation ($T_{60}$)}} \\
\quad 0.3 s & 96 & 75 & 80 & 75 & \textbf{100} & 89 & \textbf{100} & \textbf{100} & 98 & 90 & \textbf{100} & \textbf{100} & \textbf{100} & \textbf{100} & \textbf{100} & 87 & \textbf{100} & \textbf{99} & \textbf{100} & \textbf{100} \\
\quad 0.6 s & 94 & 73 & 77 & 70 & \textbf{100} & 89 & \textbf{100} & \textbf{100} & 98 & 88 & \textbf{100} & 99 & \textbf{100} & \textbf{100} & \textbf{100} & 78 & \textbf{100} & 99 & \textbf{100} & \textbf{100} \\
\multicolumn{21}{l}{\textit{Pitch (semitones)}} \\
\quad $-$2 & \textbf{67} & 11 & 8 & 37 & 1 & 0 & 0 & 0 & \textbf{97} & 77 & 1 & 0 & 0 & 0 & \textbf{95} & 4 & 1 & 0 & 0 & 0 \\
\quad $-$1 & \textbf{73} & 15 & 15 & 48 & 1 & 0 & 0 & 0 & \textbf{97} & 83 & 2 & 0 & 0 & 0 & \textbf{99} & 6 & 2 & 0 & 0 & 0 \\
\quad +1 & \textbf{78} & 13 & 11 & 46 & 1 & 0 & 0 & 0 & \textbf{97} & 86 & 1 & 0 & 0 & 0 & \textbf{100} & 7 & 1 & 0 & 0 & 0 \\
\quad +2 & \textbf{69} & 11 & 9 & 32 & 0 & 0 & 0 & 0 & \textbf{97} & 86 & 0 & 0 & 0 & 0 & \textbf{98} & 5 & 0 & 0 & 0 & 0 \\
\multicolumn{21}{l}{\textit{Speed}} \\
\quad 1.01 & 17 & 15 & 11 & 17 & 0 & 0 & \textbf{100} & 27 & 81 & 58 & 1 & 0 & \textbf{100} & \textbf{99} & 82 & 3 & 1 & 0 & \textbf{100} & 72 \\
\quad 1.02 & 10 & \textbf{13} & 8 & \textbf{16} & 0 & 0 & 0 & 0 & 36 & \textbf{62} & 1 & 0 & 0 & 4 & \textbf{49} & 3 & 1 & 0 & 0 & 1 \\
\quad 1.05 & 2 & 9 & 5 & 11 & \textbf{17} & 0 & 0 & 0 & 8 & \textbf{61} & 9 & 0 & 0 & 0 & \textbf{13} & 2 & \textbf{13} & 0 & 0 & 0 \\
\multicolumn{21}{l}{\textit{Time shift}} \\
\quad 20 ms & 25 & 8 & 3 & 4 & \textbf{100} & 89 & \textbf{100} & \textbf{100} & 74 & 8 & \textbf{100} & \textbf{100} & \textbf{100} & \textbf{100} & 58 & 1 & \textbf{100} & \textbf{100} & \textbf{100} & \textbf{100} \\
\quad 40 ms & 9 & 4 & 6 & 3 & \textbf{100} & 89 & \textbf{100} & \textbf{100} & 98 & 93 & \textbf{100} & \textbf{100} & \textbf{100} & \textbf{100} & 2 & 1 & \textbf{100} & \textbf{100} & \textbf{100} & \textbf{100} \\
\quad 80 ms & 96 & 78 & 84 & 30 & \textbf{100} & 89 & \textbf{100} & \textbf{100} & 98 & 92 & \textbf{100} & \textbf{100} & \textbf{100} & \textbf{100} & \textbf{100} & 93 & \textbf{100} & \textbf{100} & \textbf{100} & \textbf{100} \\
\multicolumn{21}{l}{\textit{Crop}} \\
\quad Half of the clip & 2 & 9 & 6 & 12 & 83 & 79 & \textbf{100} & \textbf{100} & 3 & 31 & 77 & 91 & \textbf{100} & \textbf{100} & 3 & 3 & 75 & 88 & \textbf{100} & \textbf{100} \\
\bottomrule
\end{tabular}

%% file: tables/r_a_human.tex
\begin{tabular}{lrrrrrrrrrrr}
\toprule
 & & \multicolumn{10}{c}{After eight passes} \\ \cmidrule(lr){3-12}
Method & Identity & Mimi & EnCodec & SpeechTok. & SNAC & DAC16 & DAC24 & DAC44 & Opus & MP3 & AAC \\
\midrule
Redwing & 0.5 & 0.7 & 0.5 & 0.1 & 0.6 & 0.3 & 0.8 & 0.6 & 0.6 & 0.7 & 0.7 \\
KGW & 0.7 & 0.8 & 0.5 & 0.2 & 0.7 & 0.4 & 1.0 & 0.5 & 1.0 & 0.7 & 0.4 \\
\midrule
AudioSeal & 4.3 & 42.2$^\ast$ & 79.1$^\ast$ & 2.0 & 2.7 & 4.5 & 3.5 & 4.3 & 7.0$^\ast$ & 4.4 & 3.4 \\
WavMark & 0.0 & 0.0 & 0.0 & 0.0 & 0.0 & 0.0 & 0.0 & 0.0 & 0.0 & 0.0 & 0.0 \\
Timbre & 0.7 & 0.0 & 1.6 & 0.1 & 0.0 & 0.1 & 0.0 & 0.0 & 0.1 & 0.5 & 1.2 \\
CRAW & 0.1 & 0.0 & 0.1 & 0.1 & 0.0 & 0.0 & 0.0 & 0.0 & 0.0 & 0.2 & 0.1 \\
\bottomrule
\end{tabular}

%% file: tables/r_a_fpr.tex
\begin{tabular}{lrrrrrrrrrrrr}
\toprule
 & & \multicolumn{11}{c}{After eight passes} \\ \cmidrule(lr){3-13}
Method & Identity & Native & Mimi & EnCodec & SpeechTok. & SNAC & DAC16 & DAC24 & DAC44 & Opus & MP3 & AAC \\
\midrule
\multicolumn{13}{l}{\textbf{Moshi}} \\
Redwing & 1.0 & 0.5 & -- & 0.3 & 0.3 & 0.5 & 0.0 & 0.3 & 0.8 & 0.7 & 0.2 & 0.8 \\
KGW & 1.5 & 0.3 & -- & 1.7 & 0.3 & 1.2 & 0.5 & 1.2 & 2.5 & 0.5 & 1.2 & 1.0 \\
WMAR FT & 1.3 & 0.0 & -- & 0.5 & 0.3 & 0.3 & 0.2 & 0.5 & 0.5 & 0.3 & 0.2 & 0.3 \\
WMAR FT+Augs & 1.7 & 0.5 & -- & 1.3 & 0.7 & 1.7 & 1.5 & 1.8 & 1.2 & 1.2 & 1.0 & 1.8 \\
\midrule
AudioSeal & 0.7 & 10.5$^\ast$ & -- & 75.3$^\ast$ & 1.0 & 0.7 & 1.8 & 0.8 & 0.8 & 1.8 & 0.3 & 0.7 \\
WavMark & 0.0 & 0.0 & -- & 0.0 & 0.0 & 0.0 & 0.0 & 0.0 & 0.0 & 0.0 & 0.0 & 0.0 \\
Timbre & 0.5 & 0.0 & -- & 0.0 & 0.0 & 0.2 & 0.0 & 0.0 & 0.0 & 0.2 & 0.0 & 0.7 \\
CRAW & 0.2 & 0.0 & -- & 0.3 & 0.0 & 0.0 & 0.0 & 0.0 & 0.2 & 0.2 & 0.2 & 0.2 \\
\midrule[\heavyrulewidth]\addlinespace[2pt]
\multicolumn{13}{l}{\textbf{CosyVoice3}} \\
Redwing & 0.5 & 0.7 & 0.8 & 0.5 & 1.5 & 0.7 & 1.2 & 0.5 & 1.0 & 0.7 & 0.7 & 0.5 \\
KGW & 1.8 & 1.0 & 0.3 & 0.3 & 0.5 & 0.7 & 0.7 & 1.7 & 1.3 & 1.3 & 2.2 & 0.3 \\
\midrule
AudioSeal & 1.0 & 20.8$^\ast$ & 15.1$^\ast$ & 54.1$^\ast$ & 0.8 & 0.8 & 1.0 & 1.8 & 1.3 & 2.7 & 0.8 & 0.8 \\
WavMark & 0.0 & 0.0 & 0.0 & 0.0 & 0.0 & 0.0 & 0.0 & 0.0 & 0.0 & 0.0 & 0.0 & 0.0 \\
Timbre & 0.7 & 0.2 & 0.2 & 0.7 & 0.0 & 0.5 & 0.2 & 0.0 & 0.2 & 0.3 & 1.2 & 1.3 \\
CRAW & 0.5 & 0.3 & 0.2 & 0.2 & 1.3 & 1.0 & 0.5 & 2.0 & 0.2 & 0.3 & 1.8 & 1.3 \\
\midrule[\heavyrulewidth]\addlinespace[2pt]
\multicolumn{13}{l}{\textbf{MOSS-TTS}} \\
Redwing & 0.7 & 0.7 & 0.7 & 1.3 & 1.0 & 1.0 & 1.2 & 0.7 & 0.2 & 0.2 & 0.7 & 1.8 \\
KGW & 0.3 & 0.5 & 0.2 & 0.3 & 0.3 & 0.0 & 0.2 & 0.0 & 0.8 & 0.7 & 0.2 & 0.5 \\
\midrule
AudioSeal & 0.5 & 0.8 & 4.7 & 45.9$^\ast$ & 0.0 & 0.7 & 0.5 & 0.7 & 0.8 & 1.7 & 0.8 & 0.5 \\
WavMark & 0.0 & 0.0 & 0.0 & 0.0 & 0.0 & 0.0 & 0.0 & 0.0 & 0.0 & 0.0 & 0.0 & 0.0 \\
Timbre & 0.8 & 0.2 & 0.2 & 1.2 & 0.2 & 0.0 & 0.0 & 0.2 & 0.0 & 0.2 & 0.3 & 1.2 \\
CRAW & 0.0 & 0.0 & 0.0 & 0.2 & 0.5 & 0.0 & 0.0 & 0.0 & 0.0 & 0.0 & 0.0 & 0.0 \\
\bottomrule
\end{tabular}

%% file: tables/r_a_duration.tex
\begin{tabular}{lrrrrr}
\toprule
Method & 38--50 & 51--100 & 101--150 & 151--200 & All \\
\midrule
\multicolumn{6}{l}{\textit{No attack}} \\
Redwing & 75.6\,{\scriptsize$\pm$7.6} & 98.8\,{\scriptsize$\pm$2.0} & 100.0\,{\scriptsize$\pm$1.3} & 100.0\,{\scriptsize$\pm$1.1} & 94.8\,{\scriptsize$\pm$1.8} \\
KGW & 30.0\,{\scriptsize$\pm$6.8} & 93.6\,{\scriptsize$\pm$4.1} & 99.0\,{\scriptsize$\pm$2.5} & 100.0\,{\scriptsize$\pm$1.0} & 78.5\,{\scriptsize$\pm$3.3} \\
\multicolumn{6}{l}{\textit{After eight Mimi passes}} \\
Redwing & 39.5\,{\scriptsize$\pm$8.7} & 75.8\,{\scriptsize$\pm$6.6} & 96.5\,{\scriptsize$\pm$3.3} & 100.0\,{\scriptsize$\pm$1.1} & 80.7\,{\scriptsize$\pm$3.2} \\
KGW & 1.8\,{\scriptsize$\pm$2.2} & 8.5\,{\scriptsize$\pm$4.7} & 9.6\,{\scriptsize$\pm$5.7} & 13.5\,{\scriptsize$\pm$4.9} & 8.3\,{\scriptsize$\pm$2.2} \\
\bottomrule
\end{tabular}

%% file: tables/r_a_quality_posthoc.tex
\begin{tabular}{lrrrr}
\toprule
Method & UTMOSv2 & DNSMOS Pro & NISQAv2 & NISQA-TTS \\
\midrule
\multicolumn{5}{l}{\textbf{Moshi (non-empty answers)}} \\
Unwatermarked & 3.15\,{\scriptsize$\pm$0.03} & 4.38\,{\scriptsize$\pm$0.04} & 4.64\,{\scriptsize$\pm$0.04} & 3.56\,{\scriptsize$\pm$0.05} \\
AudioSeal & \textbf{3.14}\,{\scriptsize$\pm$0.03} & \textbf{4.37}\,{\scriptsize$\pm$0.04} & \textbf{4.50}\,{\scriptsize$\pm$0.04} & 3.32\,{\scriptsize$\pm$0.05} \\
WavMark & 2.88\,{\scriptsize$\pm$0.04} & 4.26\,{\scriptsize$\pm$0.04} & \textbf{4.54}\,{\scriptsize$\pm$0.04} & 3.39\,{\scriptsize$\pm$0.04} \\
Timbre & \textbf{3.11}\,{\scriptsize$\pm$0.03} & \textbf{4.31}\,{\scriptsize$\pm$0.04} & \textbf{4.55}\,{\scriptsize$\pm$0.05} & 3.30\,{\scriptsize$\pm$0.04} \\
CRAW & 2.92\,{\scriptsize$\pm$0.04} & 4.25\,{\scriptsize$\pm$0.05} & 3.96\,{\scriptsize$\pm$0.07} & \textbf{3.57}\,{\scriptsize$\pm$0.05} \\
\midrule[\heavyrulewidth]\addlinespace[2pt]
\multicolumn{5}{l}{\textbf{CosyVoice3}} \\
Unwatermarked & 2.91\,{\scriptsize$\pm$0.03} & 4.30\,{\scriptsize$\pm$0.03} & 3.78\,{\scriptsize$\pm$0.07} & 3.52\,{\scriptsize$\pm$0.05} \\
AudioSeal & \textbf{2.93}\,{\scriptsize$\pm$0.03} & \textbf{4.33}\,{\scriptsize$\pm$0.03} & \textbf{3.98}\,{\scriptsize$\pm$0.05} & \textbf{3.62}\,{\scriptsize$\pm$0.05} \\
WavMark & \textbf{2.91}\,{\scriptsize$\pm$0.03} & 4.24\,{\scriptsize$\pm$0.03} & 3.70\,{\scriptsize$\pm$0.06} & \textbf{3.56}\,{\scriptsize$\pm$0.05} \\
Timbre & \textbf{2.90}\,{\scriptsize$\pm$0.03} & \textbf{4.29}\,{\scriptsize$\pm$0.02} & 3.75\,{\scriptsize$\pm$0.06} & 3.48\,{\scriptsize$\pm$0.05} \\
CRAW & 2.67\,{\scriptsize$\pm$0.03} & 4.02\,{\scriptsize$\pm$0.04} & 3.30\,{\scriptsize$\pm$0.06} & 3.27\,{\scriptsize$\pm$0.05} \\
\midrule[\heavyrulewidth]\addlinespace[2pt]
\multicolumn{5}{l}{\textbf{MOSS-TTS}} \\
Unwatermarked & 2.85\,{\scriptsize$\pm$0.04} & 4.20\,{\scriptsize$\pm$0.03} & 3.61\,{\scriptsize$\pm$0.07} & 3.28\,{\scriptsize$\pm$0.05} \\
AudioSeal & \textbf{2.86}\,{\scriptsize$\pm$0.04} & \textbf{4.22}\,{\scriptsize$\pm$0.03} & \textbf{3.75}\,{\scriptsize$\pm$0.06} & \textbf{3.30}\,{\scriptsize$\pm$0.06} \\
WavMark & \textbf{2.91}\,{\scriptsize$\pm$0.03} & \textbf{4.21}\,{\scriptsize$\pm$0.03} & 3.59\,{\scriptsize$\pm$0.06} & \textbf{3.35}\,{\scriptsize$\pm$0.05} \\
Timbre & \textbf{2.86}\,{\scriptsize$\pm$0.04} & \textbf{4.18}\,{\scriptsize$\pm$0.03} & 3.54\,{\scriptsize$\pm$0.06} & \textbf{3.30}\,{\scriptsize$\pm$0.05} \\
CRAW & 2.57\,{\scriptsize$\pm$0.03} & 3.86\,{\scriptsize$\pm$0.04} & 3.17\,{\scriptsize$\pm$0.06} & 3.01\,{\scriptsize$\pm$0.05} \\
\bottomrule
\end{tabular}

%% file: tables/r_a_quality_delta.tex
\begin{tabular}{lrrrr}
\toprule
Method & $\Delta$UTMOSv2 & $\Delta$DNSMOS Pro & $\Delta$NISQAv2 & $\Delta$NISQA-TTS \\
\midrule
\multicolumn{5}{l}{\textbf{Moshi (prompts non-empty for all)}} \\
Redwing & $-0.09$\,{\scriptsize$\pm$0.04} & $+0.10$\,{\scriptsize$\pm$0.05} & $0.00$\,{\scriptsize$\pm$0.06} & $-0.03$\,{\scriptsize$\pm$0.06} \\
KGW & $-0.09$\,{\scriptsize$\pm$0.06} & $-0.04$\,{\scriptsize$\pm$0.07} & $-0.15$\,{\scriptsize$\pm$0.08} & $-0.14$\,{\scriptsize$\pm$0.07} \\
WMAR FT$^\dagger$ & $-0.07$\,{\scriptsize$\pm$0.05} & $-0.04$\,{\scriptsize$\pm$0.07} & $-0.12$\,{\scriptsize$\pm$0.07} & $-0.16$\,{\scriptsize$\pm$0.07} \\
WMAR FT+Augs$^\dagger$ & $-0.11$\,{\scriptsize$\pm$0.05} & $-0.05$\,{\scriptsize$\pm$0.07} & $-0.15$\,{\scriptsize$\pm$0.08} & $-0.15$\,{\scriptsize$\pm$0.07} \\
\midrule
AudioSeal & $-0.01$\,{\scriptsize$\pm$0.02} & $-0.01$\,{\scriptsize$\pm$0.01} & $-0.15$\,{\scriptsize$\pm$0.01} & $-0.24$\,{\scriptsize$\pm$0.02} \\
WavMark & $-0.27$\,{\scriptsize$\pm$0.02} & $-0.12$\,{\scriptsize$\pm$0.01} & $-0.10$\,{\scriptsize$\pm$0.02} & $-0.17$\,{\scriptsize$\pm$0.03} \\
Timbre & $-0.03$\,{\scriptsize$\pm$0.02} & $-0.06$\,{\scriptsize$\pm$0.01} & $-0.09$\,{\scriptsize$\pm$0.02} & $-0.25$\,{\scriptsize$\pm$0.03} \\
CRAW & $-0.23$\,{\scriptsize$\pm$0.03} & $-0.12$\,{\scriptsize$\pm$0.02} & $-0.68$\,{\scriptsize$\pm$0.05} & $+0.01$\,{\scriptsize$\pm$0.04} \\
\midrule[\heavyrulewidth]\addlinespace[2pt]
\multicolumn{5}{l}{\textbf{CosyVoice3}} \\
Redwing & $+0.01$\,{\scriptsize$\pm$0.02} & $-0.01$\,{\scriptsize$\pm$0.02} & $-0.01$\,{\scriptsize$\pm$0.05} & $0.00$\,{\scriptsize$\pm$0.04} \\
KGW & $0.00$\,{\scriptsize$\pm$0.02} & $-0.02$\,{\scriptsize$\pm$0.02} & $-0.05$\,{\scriptsize$\pm$0.05} & $-0.01$\,{\scriptsize$\pm$0.04} \\
\midrule
AudioSeal & $+0.03$\,{\scriptsize$\pm$0.02} & $+0.02$\,{\scriptsize$\pm$0.01} & $+0.20$\,{\scriptsize$\pm$0.03} & $+0.11$\,{\scriptsize$\pm$0.02} \\
WavMark & $0.00$\,{\scriptsize$\pm$0.02} & $-0.07$\,{\scriptsize$\pm$0.02} & $-0.09$\,{\scriptsize$\pm$0.04} & $+0.04$\,{\scriptsize$\pm$0.03} \\
Timbre & $0.00$\,{\scriptsize$\pm$0.02} & $-0.01$\,{\scriptsize$\pm$0.01} & $-0.03$\,{\scriptsize$\pm$0.03} & $-0.03$\,{\scriptsize$\pm$0.02} \\
CRAW & $-0.23$\,{\scriptsize$\pm$0.02} & $-0.28$\,{\scriptsize$\pm$0.02} & $-0.48$\,{\scriptsize$\pm$0.05} & $-0.25$\,{\scriptsize$\pm$0.04} \\
\midrule[\heavyrulewidth]\addlinespace[2pt]
\multicolumn{5}{l}{\textbf{MOSS-TTS}} \\
Redwing & $-0.12$\,{\scriptsize$\pm$0.03} & $-0.04$\,{\scriptsize$\pm$0.02} & $-0.01$\,{\scriptsize$\pm$0.05} & $-0.04$\,{\scriptsize$\pm$0.04} \\
KGW & $-0.09$\,{\scriptsize$\pm$0.03} & $-0.01$\,{\scriptsize$\pm$0.02} & $+0.04$\,{\scriptsize$\pm$0.04} & $-0.01$\,{\scriptsize$\pm$0.04} \\
\midrule
AudioSeal & $+0.01$\,{\scriptsize$\pm$0.02} & $+0.02$\,{\scriptsize$\pm$0.01} & $+0.15$\,{\scriptsize$\pm$0.03} & $+0.02$\,{\scriptsize$\pm$0.02} \\
WavMark & $+0.05$\,{\scriptsize$\pm$0.02} & $+0.01$\,{\scriptsize$\pm$0.02} & $-0.02$\,{\scriptsize$\pm$0.03} & $+0.07$\,{\scriptsize$\pm$0.02} \\
Timbre & $+0.01$\,{\scriptsize$\pm$0.02} & $-0.03$\,{\scriptsize$\pm$0.01} & $-0.06$\,{\scriptsize$\pm$0.02} & $+0.01$\,{\scriptsize$\pm$0.02} \\
CRAW & $-0.28$\,{\scriptsize$\pm$0.02} & $-0.35$\,{\scriptsize$\pm$0.03} & $-0.44$\,{\scriptsize$\pm$0.05} & $-0.27$\,{\scriptsize$\pm$0.04} \\
\bottomrule
\end{tabular}

%% file: tables/r_t3_ablation.tex
\begin{tabular}{lrrrrrrrr}
\toprule
 & & & & & \multicolumn{4}{c}{After eight passes} \\ \cmidrule(lr){6-9}
Variant & $\delta$ & KL & UTMOSv2 & Identity & Mimi & EnCodec & SpeechTok. & DAC16 \\
\midrule
Unwatermarked & -- & 0 & 3.15\,{\scriptsize$\pm$0.03} & -- & -- & -- & -- & -- \\
KGW & 2 & 1.73 & 3.06\,{\scriptsize$\pm$0.04} & 78.5\,{\scriptsize$\pm$3.3} & 8.3\,{\scriptsize$\pm$2.2} & 16.5\,{\scriptsize$\pm$3.0} & 5.3\,{\scriptsize$\pm$1.8} & 12.5\,{\scriptsize$\pm$2.6} \\
\midrule
Redwing & 0.7 & 1.52 & 3.05\,{\scriptsize$\pm$0.03} & \textbf{94.8}\,{\scriptsize$\pm$1.8} & \textbf{80.7}\,{\scriptsize$\pm$3.2} & \textbf{85.0}\,{\scriptsize$\pm$2.9} & \textbf{77.5}\,{\scriptsize$\pm$3.3} & 34.2\,{\scriptsize$\pm$3.8} \\
\multicolumn{9}{l}{\textit{Basis}} \\
\quad Random basis & 0.7 & 1.60 & 3.12\,{\scriptsize$\pm$0.03} & \textbf{92.0}\,{\scriptsize$\pm$2.2} & 52.3\,{\scriptsize$\pm$4.0} & 36.0\,{\scriptsize$\pm$3.8} & 24.5\,{\scriptsize$\pm$3.4} & 45.5\,{\scriptsize$\pm$4.0} \\
\quad Transition basis ($P$) & 0.65 & 1.55 & 2.99\,{\scriptsize$\pm$0.03} & \textbf{92.7}\,{\scriptsize$\pm$2.1} & 68.5\,{\scriptsize$\pm$3.7} & 73.8\,{\scriptsize$\pm$3.5} & 68.5\,{\scriptsize$\pm$3.7} & \textbf{56.5}\,{\scriptsize$\pm$4.0} \\
\multicolumn{9}{l}{\textit{Solved functions}} \\
\quad Random function, $h=g$ & 1 & 1.64 & 3.10\,{\scriptsize$\pm$0.03} & 79.7\,{\scriptsize$\pm$3.2} & 16.8\,{\scriptsize$\pm$3.0} & 19.5\,{\scriptsize$\pm$3.2} & 6.8\,{\scriptsize$\pm$2.0} & 14.8\,{\scriptsize$\pm$2.8} \\
\quad Detection with $h:=g$$^\ddagger$ & 0.7 & 1.52 & \textit{same} & \textbf{93.3}\,{\scriptsize$\pm$2.0} & \textbf{76.5}\,{\scriptsize$\pm$3.4} & \textbf{80.0}\,{\scriptsize$\pm$3.2} & \textbf{75.8}\,{\scriptsize$\pm$3.4} & \textbf{54.2}\,{\scriptsize$\pm$4.0} \\
\multicolumn{9}{l}{\textit{Amplitude bounds}} \\
\quad No bounds ($\kappa=\infty$) & 0.07 & 0.91 & 2.71\,{\scriptsize$\pm$0.05} & 36.3\,{\scriptsize$\pm$3.8} & 18.7\,{\scriptsize$\pm$3.1} & 25.2\,{\scriptsize$\pm$3.5} & 12.0\,{\scriptsize$\pm$2.6} & 15.2\,{\scriptsize$\pm$2.9} \\
\bottomrule
\end{tabular}